\documentclass[11pt]{article}
 \usepackage[utf8]{inputenc}
 \usepackage{fullpage}

\usepackage{fullpage}
\usepackage{amsmath,amsthm,amsfonts,dsfont}
\usepackage{amssymb,latexsym,graphicx}
\usepackage{mathtools}
\usepackage[breaklinks]{hyperref}
\usepackage[capitalize,nameinlink,noabbrev,nosort]{cleveref}
\usepackage{xcolor}
\definecolor{etonblue}{rgb}{0.59, 0.78, 0.64}
\hypersetup{
	breaklinks=true,
	colorlinks,
	linkcolor={green!40!blue},
	citecolor={etonblue!80!black},
	urlcolor={red!75!black}
}

\usepackage{breakurl}

\newtheorem{theorem}{Theorem}[section]

\newtheorem{proposition}[theorem]{Proposition}
\newtheorem{lemma}[theorem]{Lemma}

\newtheorem{claim}[theorem]{Claim}

\newtheorem{corollary}[theorem]{Corollary}
\newtheorem{definition}[theorem]{Definition}

\newcommand{\R}{\ensuremath{\mathbb{R}}}
\newcommand{\Z}{\ensuremath{\mathbb{Z}}}
\newcommand{\Q}{\mathbb{Q}}

\newcommand{\lat}{\mathcal{L}}
\newcommand{\M}{\mathcal{M}}

\newcommand{\eps}{\varepsilon} 
\renewcommand{\epsilon}{\varepsilon}
\newcommand{\poly}{\mathrm{poly}}

\newcommand{\rank}{\mathrm{rank}}

\DeclareMathOperator{\spn}{span}

\renewcommand{\vec}[1]{\ensuremath{\boldsymbol{#1}}}
\newcommand{\basis}{\ensuremath{\mathbf{B}}}
\newcommand{\gs}[1]{\ensuremath{\widetilde{\vec{#1}}}}

\DeclarePairedDelimiter\inner{\langle}{\rangle}
\DeclarePairedDelimiter\floor{\lfloor}{\rfloor}
\DeclarePairedDelimiter\ceil{\lceil}{\rceil}

\mathtoolsset{centercolon}%

\usepackage{enumitem}
\usepackage{mathrsfs}
\usepackage{xcolor}
\usepackage[hyperpageref]{backref}
\usepackage{tikz}
\usetikzlibrary{arrows.meta}
\usepackage[all]{nowidow}

\newcommand{\A}{\mathcal{A}}

\renewcommand{\v}{\vec{v}}
\newcommand{\w}{\vec{w}}
\newcommand{\y}{\vec{y}}
\renewcommand{\b}{\vec{b}}
\renewcommand{\basis}{\mathbf{B}}
\renewcommand{\inner}[1]{\langle #1 \rangle}
\DeclareMathOperator{\Span}{span}

\renewcommand{\hat}{\widehat}

\usepackage[font=small,labelfont=bf]{caption}

\title{\vspace{-4em}Recursive lattice reduction---\\ A framework for finding short lattice vectors}

\author{
Divesh Aggarwal\\National University of Singapore\\ \texttt{divesh@comp.nus.edu.sg} \and 
Thomas Espitau\\PQShield\\\texttt{t.espitau@gmail.com}\and
Spencer Peters\\Cornell University\\ \texttt{sp2473@cornell.edu}  \and 
Noah Stephens-Davidowitz\\Cornell University\\\texttt{noahsd@gmail.com}
}

\date{\today}

\begin{document}

\pagenumbering{roman}

\maketitle

\begin{abstract}
    We propose a new framework called recursive lattice reduction for finding short non-zero vectors in a lattice or for finding dense sublattices of a lattice. At a high level, the framework works by recursively searching for dense sublattices of dense sublattices (or their duals) with progressively lower rank. Eventually, the procedure encounters a recursive call on a lattice $\lat$ with relatively low rank $k$, at which point we can simply use a known algorithm to find a shortest non-zero vector in $\lat$.

    We view this new framework as complementary to basis reduction algorithms, which similarly work to reduce an $n$-dimensional lattice problem with some approximation factor $\gamma$ to a lower-dimensional exact lattice problem in some lower dimension $k$, with a tradeoff between $\gamma$, $n$, and $k$. Our framework provides an alternative and arguably simpler perspective. For example, 
    our algorithms can be described at a high level without explicitly referencing any specific basis of the lattice, the Gram-Schmidt orthogonalization, or even projection (though, of course, concrete implementations of algorithms in this framework will likely make use of such things, and we discuss two such implementations as proofs of concept). 

    We present a number of specific instantiations of our framework to illustrate its usefulness. Our main concrete result is an efficient reduction that matches the tradeoff between $\gamma$, $n$, and $k$ achieved by the best-known basis reduction algorithms (in terms of the Hermite factor, up to low-order terms) across all parameter regimes. In fact, this reduction also can be used to find dense \emph{sublattices} with any rank $\ell$ satisfying $\min\{\ell,n-\ell\} \leq n-k+1$, using only an oracle for SVP (or even just Hermite SVP) in $k$ dimensions, even with slightly better parameters than what was known using basis reduction.
    
    We also show an extremely simple reduction that achieves the same tradeoff for finding short vectors (up to low order terms) in quasipolynomial time, and a reduction from the problem of finding dense sublattices of a high-dimensional lattice to  the problem of finding dense sublattices of lower-dimensional lattices. Finally, we present a simple efficient automated search procedure for finding algorithms in this framework that (provably) achieve better approximations with fewer oracle calls.
\end{abstract}
\thispagestyle{empty}
\newpage
\setcounter{tocdepth}{2}
\tableofcontents
\newpage
\pagenumbering{arabic}

\section{Introduction}

A lattice $\lat \subset \R^d$ is the set of all integer linear combinations of $n$ linearly independent basis vectors $\basis = (\vec{b}_1,\ldots, \vec{b}_n) \in \R^{d \times n}$, i.e.,
\[
    \lat := \{ z_1 \vec{b}_1 + z_2 \vec{b}_2 + \cdots + z_n \vec{b}_n \ : \ z_i \in \Z\}
    \; .
\]
We call $n$ the \emph{rank} of the lattice and $d$ the \emph{dimension} (or sometimes the \emph{ambient dimension}). We often implicitly assume that $d =n$ (which is essentially without loss of generality, by identifying $\spn(\lat)$ with $\R^n$).

The most important geometric quantities associated with a lattice $\lat \subset \R^d$ are the length of a shortest non-zero lattice vector, $\lambda_1(\lat) := \min_{\vec{y} \in \lat_{ \neq \vec0}} \|\vec{y}\|$, and the determinant, $\det(\lat) := \det(\basis^T \basis)^{1/2}$. (Here and throughout this paper, $\|\vec{x}\| := (x_1^2 + \cdots + x_d^2)^{1/2}$ means the Euclidean norm.) The determinant is best viewed as a measure of the \emph{global density} of a lattice, with lattices with smaller determinant being more dense, while $\lambda_1(\lat)$ is similarly a measure of the \emph{local density}. These two quantities are related by \emph{Hermite's constant}, which is defined as
\[
    \delta_n := \sup_{\lat \subset \R^d: \; \rank(\lat) = n} \big(\lambda_1(\lat)/\det(\lat)^{\frac 1n}\big)^2
    \; .
\]
This is a very well studied quantity.
Indeed, Minkowski's celebrated first theorem tells us that $\delta_n \leq O(n)$, and the Minkowski-Hlawka theorem tells us that $\delta_n \geq \Omega(n)$. So, we know that $\delta_n = \Theta(n)$.

In this work, we are interested in the $\gamma$-Hermite Shortest Vector Problem ($\gamma$-HSVP) for $\gamma = \gamma(n) \geq \sqrt{\delta_n}$. The input in $\gamma$-HSVP is a basis for a lattice $\lat \subset \R^n$, and the goal is to find a non-zero lattice vector $\vec{y} \in \lat_{\neq\vec0}$ such that $\|\vec{y}\| \leq \gamma \cdot\det(\lat)^{1/n}$.\footnote{In \cref{sec:HSVP}, we discuss the relationship between HSVP and the Shortest Vector Problem, in which the goal is to find a non-zero lattice vector whose length is within a certain factor of the shortest such vector in the lattice.} The parameter $\gamma$ is called the \emph{approximation factor} or sometimes the Hermite factor. Notice that a solution is guaranteed to exist if and only if $\gamma \geq \sqrt{\delta_n}$.

This problem is central to lattice-based cryptography, which is in the process of widespread deployment~\cite{NIST2022}. In particular, the best known attacks on lattice-based cryptography essentially work via reduction to $\gamma$-HSVP for $\gamma = \poly(n)$.  Thus,  \emph{precise} estimates of the time complexity of $\gamma$-HSVP are necessary for assessing the security of these schemes. (See, e.g.,~\cite{albrechtEstimateAllLWE2018}.)

Algorithms for $\gamma$-HSVP have a \emph{very} rich history.\footnote{Most of the algorithms that we list here were originally presented as algorithms for \emph{the Shortest Vector Problem} (SVP). We are describing them as algorithms for HSVP, since this is what interests us (and, indeed, what is typically relevant to cryptographers). See \cref{sec:HSVP} for more discussion about the relationship between SVP and HSVP.} In the hardest possible case when $\gamma = \sqrt{\delta_n}$, the fastest known algorithm for HSVP runs in $2^{\Theta(n)}$-time~\cite{ajtaiSieveAlgorithmShortest2001}, with a long line of work improving on the constant in the exponent~\cite{pujolSolvingShortestLattice2009,micciancioFasterExponentialTime2010,micciancioDeterministicSingleExponential2010,aggarwalSolvingShortestVector2015} (and another long line of work on heuristic algorithms~\cite{nguyenSieveAlgorithmsShortest2008,laarhovenSievingShortestVectors2015,beckerNewDirectionsNearest2016}, and yet another long line of work on algorithms that do not quite achieve $\gamma = \sqrt{\delta_n}$ but do come very close, e.g., achieving $\gamma \leq \sqrt{\delta_n} \cdot \poly(\log n)$~\cite{liuShortestLatticeVectors2011,weiFindingShortestLattice2015,aggarwalFasterSievingAlgorithm2019,ALSTimeAlgorithmSqrtn2021}). The case of larger $\gamma$ (say $\gamma > n^{1/2+\eps}$) is more relevant to the current work, and we discuss it in depth below. (See also \cite{BenComplexityShortestVector2023} for a recent survey on the complexity of SVP.)

\subsection{Basis reduction algorithms}

The first non-trivial\footnote{By specifying ``non-trivial'' algorithms for $\gamma$-HSVP for $\gamma \gg \sqrt{\delta_n}$ here and elsewhere, we mean to exclude algorithms that work by simply solving $\gamma'$-HSVP with $\gamma' \approx \sqrt{\delta_n}$ on the input lattice. E.g., it is ``trivial'' to solve, say, $2^n$-HSVP by simply running, say, the $\sqrt{\delta_n}$-HSVP algorithm from~\cite{ajtaiSieveAlgorithmShortest2001} on the input lattice. Of course, we do not mean to suggest that the \cite{ajtaiSieveAlgorithmShortest2001} algorithm is trivial or that it is trivial to solve $\sqrt{\delta_n}$-HSVP!} algorithm for $\gamma$-HSVP for larger $\gamma \gg \sqrt{\delta_n}$ was the celebrated LLL algorithm by Lenstra, Lenstra, and Lov{\'a}sz~\cite{LLL82}. Their algorithm runs in polynomial time and achieves an approximation factor of $\gamma = 2^{O(n)}$. Indeed, forty years later, essentially all known non-trivial algorithms for $\gamma$-HSVP for large $\gamma \gg \sqrt{\delta_n}$ still use the technique introduced by LLL: basis reduction. So, in some (imprecise) sense, basis reduction is the only known framework for solving approximate lattice problems, and the community has spent the past forty years generating a beautiful body of literature devoted to perfecting this technique~\cite{schnorrHierarchyPolynomialTime1987,schnorrLatticeBasisReduction1994,gamaPredictingLatticeReduction2008,gamaFindingShortLattice2008,hanrotAnalyzingBlockwiseLattice2011,micciancioPracticalPredictableLattice2016}. (See \cite{walterLatticeBlogReduction} for a recent popular survey.)

At a high level, basis reduction algorithms work to progressively find a shorter basis for the input lattice $\lat \subset \R^n$ by solving exact SVP instances on carefully chosen lattices with rank $k \geq 2$, where $k < n$ is known as the block size. Since the fastest known algorithms for $\sqrt{\delta_k}$-HSVP on lattices with rank $k$ run in time $2^{\Theta(k)}$, the specific relationship between the block size $k$ and approximation factor $\gamma$ is quite important. Basis reduction achieves a smooth tradeoff, yielding an efficient reduction from $\gamma_{\mathsf{BR}}$-HSVP on lattices with rank $n$ to $\sqrt{\delta_k}$-HSVP on lattices with rank $k$ that achieves 
\[\gamma_{\mathsf{BR}} \approx \delta_k^{\frac{n-1}{2(k-1)}} \approx k^{\frac n{2k}}
\; .
\] 
See, e.g.,~\cite{gamaFindingShortLattice2008,micciancioPracticalPredictableLattice2016,aggarwalSlideReductionRevisited2020,walterLatticeBlogReduction,ALSTimeAlgorithmSqrtn2021}. (For cryptography, we are mostly concerned with the case when $k$ is rather large, e.g., $k = Cn$ for some not-too-small constant $C > 0$.)

In more detail (which is not actually necessary for understanding the rest of this paper), basis reduction algorithms work by attempting to find a ``good'' basis of the given lattice (and in particular, a basis whose first vector $\vec{b}_1$ is quite short) by manipulating the Gram-Schmidt orthogonalization of the basis, $\gs{\vec{b}}_1 := \vec{b}_1, \gs{\vec{b}}_2 := \Pi_{\{\vec{b}_1\}^\perp}(\vec{b}_2), \ldots, \gs{\vec{b}}_n := \Pi_{\{\vec{b}_1,\ldots, \vec{b}_{n-1}\}^\perp}(\vec{b}_n)$, where $\Pi_{\{\vec{b}_1,\ldots, \vec{b}_i\}^\perp}$ represents projection onto the subspace orthogonal to $\vec{b}_1,\ldots, \vec{b}_i$. The idea behind all basis reduction algorithms is to make earlier Gram-Schmidt vectors $\gs{\vec{b}}_i$ shorter at the expense of making later Gram-Schmidt vectors $\gs{\vec{b}}_j$ for $j > i$ longer.\footnote{The product of the lengths of the Gram-Schmidt vectors is the determinant of the lattice, and therefore cannot be changed without changing the lattice. So, if one wishes to make earlier Gram-Schmidt vectors shorter, one \emph{must} lengthen later Gram-Schmidt vectors.} In particular, one hopes to make the first Gram-Schmidt vector $\gs{\vec{b}}_1 = \vec{b}_1$ short, since it is actually a lattice vector (while all other Gram-Schmidt vectors are \emph{projections} of lattice vectors, which are typically not themselves lattice vectors). To make $\gs{\vec{b}}_i$ shorter relative to later Gram-Schmidt vectors, basis reduction algorithms rely on the fact that $\gs{\vec{b}}_i$ is contained in a lattice with rank $k$ whose determinant is $\|\gs{\vec{b}}_i\| \cdots \|\gs{\vec{b}}_{i+k-1}\|$. Using an oracle for $\sqrt{\delta_k}$-HSVP on rank-$k$ lattices allows us to find a non-zero vector in this lattice that is short relative to the geometric mean of $\|\gs{\vec{b}}_{i+1}\|, \ldots, \|\gs{\vec{b}}_{i+k-1}\|$. We can then substitute that vector for $\gs{\vec{b}}_i$. By doing this many different times with $i=1,\ldots, n-k+1$, we can hope to eventually guarantee that $\|\gs{\vec{b}}_1\|, \|\gs{\vec{b}}_2\|,\ldots, \|\gs{\vec{b}}_k\|$ are not that large relative to the determinant of the lattice. And, with one more call to an HSVP oracle on the lattice generated by $\vec{b}_1,\ldots, \vec{b}_k$ (whose determinant is $\|\gs{\vec{b}}_1\|\|\gs{\vec{b}}_2\| \cdots \|\gs{\vec{b}}_k\|$), we can find a relatively short lattice vector.

So, the high-level goal of basis reduction algorithms is to find a rank-$k$ sublattice $\lat' \subset \lat$ with low determinant (the lattice generated by $\vec{b}_1,\ldots, \vec{b}_k$ above) and then to use a HSVP oracle on rank-$k$ lattices to find a short non-zero vector in $\lat'$ (though basis reduction algorithms are not typically described in this way). However, the details matter quite a bit in basis reduction, and they get complicated quickly. As a result, analyzing these algorithms can be rather difficult, because the effect of each oracle call on the lengths of the different Gram-Schmidt vectors is difficult to control. 

In fact, the best basis reduction algorithms are \emph{heuristic}, meaning that we do not know how to prove their correctness~\cite{CNBKZBetterLattice2011}. \emph{And}, our best way of analyzing the behavior of these algorithms is via sophisticated computer simulations~\cite{gamaPredictingLatticeReduction2008,CNBKZBetterLattice2011,bai2018measuring}.\footnote{Gama and Nguyen's elegant slide reduction algorithm~\cite{gamaFindingShortLattice2008} is a bit of an exception---i.e., it is a basis reduction algorithm that can be analyzed quite simply. However, its performance in practice still lags behind heuristic basis reduction algorithms whose behavior is not nearly so well understood. One can modify  slide reduction to improve its performance~\cite{WalConvergenceSlidetypeReductions2021}, but at the expense of losing the simple proof of correctness.}

\section{An overview of recursive lattice reduction}

In this work, we present a new framework for reductions from $\gamma$-HSVP on rank-$n$ lattices to $\sqrt{\delta_k}$-HSVP on rank-$k$ lattices, which we call \emph{recursive lattice reduction}.\footnote{The word ``reduction'' has two meanings here, as it often does in the literature on basis reduction. First, it simply means a ``reduction'' in the traditional computer science sense: an algorithm for one problem that requires access to an oracle that solves another problem---in this case, HSVP with lower rank. Second, it means ``reduction'' in a sense that is specific to the context of lattices; ``reduction'' in this sense describes a procedure that slowly \emph{reduces} the lengths of lattice vectors (or the density of sublattices). This double meaning can be rather confusing at times---e.g., ``basis reduction algorithms'' should perhaps be more accurately be called ``basis reduction reductions,'' but this terminology is of course quite cumbersome. Here, we are using the term ``reduction'' rather ambiguously, in analogy with basis reduction. It is an open question whether the authors will regret our use of this ambiguous terminology.} 

This new framework maintains the high-level idea that, in order to find a short vector in a high-rank lattice, one should first find a dense lower-rank sublattice $\lat' \subset \lat$ and then find a short vector in $\lat'$.  And, our framework therefore is rather similar to basis reduction.
However, in our framework this high-level idea is the explicit guiding principle behind both the reduction and its analysis (as we will see shortly), while in basis reduction, this high-level idea is typically confined to the analysis.

And, our framework is in some sense more flexible than basis reduction. E.g., we do not require that $\lat'$ has rank precisely $k$. And, we do not even require that the reduction explicitly works with a basis (or an essentially equivalent representation, such as a Gram matrix or a flag)---let alone the Gram-Schmidt orthogonalization. (This is why we are careful to call this recursive \emph{lattice} reduction, and \emph{not} to call it \emph{basis} reduction.) This independence from the specific representation of the lattice makes our reductions very easy to describe at a high level. (Of course, any actual concrete implementation of our reductions will have to represent the lattice in some way---probably with a basis, and perhaps even using the Gram-Schmidt orthogonalization to perform basic operations on the lattice. But, our framework is agnostic to such choices, and we instead view the specific choice of representation as an implementation detail.)

In more detail, the basic idea behind our framework is simply the following. Our reductions $\A(\lat, \mathsf{aux})$ take as input a lattice $\lat \subset \R^n$ and some simple auxiliary information $\mathsf{aux}$. If the lattice $\lat$ has sufficiently small rank $n \leq k$, then we of course use an oracle call to find a short vector. Otherwise, the reduction makes (at least) two recursive calls. The first recursive call is used to find a sublattice $\lat' \subset \lat$ with rank less than $n$ and relatively small determinant. The next recursive call is simply $\A(\lat', \mathsf{aux}')$.

In our concrete instantiations of the framework, we use the auxiliary input $\mathsf{aux}$ to do two things: (1) to keep track of the rank of the sublattice that we are currently looking for; and (2) to control the running time of the reduction (and its subsequent recursive calls---specifically).

As we will see, this is a sufficiently modular and simple paradigm to allow for computer-aided search of the space of reductions in this framework. For example, one can write a simple computer program that takes as input a rank $n$, block size $k$, and a bound $T$ on the number of allowed oracle calls. The program will output a description of a reduction from $\gamma$-HSVP on rank-$n$ lattices to $\sqrt{\delta_k}$-HSVP on rank-$k$ lattices, together with an approximation factor $\gamma$ that the reduction is \emph{guaranteed} to achieve. Furthermore, the value of $\gamma$ will be optimal, in some (admittedly rather weak) sense. (Specifically, it will be optimal among all reductions of a certain form.)

Below, we build up the details of our framework slowly by introducing progressively more sophisticated examples. We first develop a simple and natural (but still interesting and novel) reduction in the framework, in \cref{sec:SVPtoSVP}. Then in \cref{sec:intro-DSP-to-DSP,sec:intro-to-hsvp}, we provide a high-level overview of the more sophisticated reductions whose details we leave for the sequel. \cref{sec:intro-to-hsvp} (and the corresponding formal details in \cref{sec:DSP-to-SVP}) in particular shows a reduction from the problem of finding dense sublattices of rank-$n$ lattices to the problem of finding a short vector in rank-$k$ lattices. The specific time-approximation tradeoff that we achieve at least matches the state-of-the-art in basis reduction for all parameter regimes (up to low-order terms), and outperforms basis reduction in some regimes, though this is not our main focus.

We note that throughout this paper, our purpose is to illustrate what is possible to do with this framework, rather than to optimize parameters. And we are particularly uninterested in optimizing the precise running time of our reductions. We have therefore made little attempt to, e.g., choose optimal constants. And, in order to make our presentation as simple as possible, we have resisted the urge to include optimizations that one absolutely would wish to include if one were to implement these algorithms. (As one rather extreme example, the astute reader might notice that all of our reductions as described run the LLL algorithm many more times than is truly necessary---sometimes running the LLL algorithm up to $n$ times on the same lattice! Again, we have resisted the urge to optimize away such inefficiencies.)

\subsection{A first example}
\label{sec:SVPtoSVP}

To illustrate the new framework, consider the following simple (and novel) idea for a reduction from $\gamma$-HSVP on a rank-$n$ lattice to $\sqrt{\delta_k}$-HSVP on rank-$k$ lattices. Indeed, the reduction is sufficiently simple that we will provide a complete description and analysis here in this overview.

The idea is for the reduction to take as input a lattice with rank $n$ and to find a sublattice $\lat' \subset \lat$ with rank $n-1$ and relatively small determinant. The reduction then calls itself recursively on $\lat'$, therefore finding a short non-zero vector in $\lat'$, which is of course also a short non-zero vector in $\lat$.

To make the above precise, we first explain how to find a dense sublattice $\lat' \subset \lat$ with rank $n-1$ of $\lat$. We do this using duality and recursion. In particular, it is a basic fact that the \emph{dual} $\lat^*$ of a lattice $\lat$ is itself a lattice with $\det(\lat^*) = 1/\det(\lat)$ with the property that every (primitive) sublattice $\lat'$ of $\lat$ with rank $n-1$ is simply the intersection $\lat' = \lat \cap \vec{w}^\perp$ of $\lat$ with the subspace $\vec{w}^\perp$ orthogonal to some non-zero dual vector $\vec{w} \in \lat^*$. (See \cref{sec:prelims-dual} for a formal definition of the dual lattice $\lat^*$ and of primitivity, and for a discussion of many related properties.) Furthermore, $\det(\lat') = \|\vec{w}\| \det(\lat)$ (assuming that $\vec{w}$ is primitive). Therefore, finding a dense sublattice $\lat' \subset \lat$ with rank $n-1$ is \emph{equivalent} to finding a short non-zero dual vector $\vec{w}$.  (This basic fact is used quite a bit already in basis reduction algorithms, e.g., in~\cite{gamaFindingShortLattice2008,micciancioPracticalPredictableLattice2016}.)

So, to find such a dense sublattice $\lat'$, it suffices to find a short vector in the dual. We would like to do this by simply calling our reduction recursively on $\lat^*$. However, this would obviously lead to an infinite loop! Indeed, to find a short vector in $\lat$, a naive implementation of this procedure would attempt to find a short vector in $\lat^*$ by attempting to find a short vector in $(\lat^*)^* = \lat$, etc! The solution is to add a \emph{depth parameter} $\tau$ as auxiliary input to the reduction. Intuitively, when the depth parameter is larger, the reduction takes more time but achieves a better approximation factor $\gamma$ (i.e., finds shorter vectors). The hope is that the recursive call on $\lat^*$ can afford a worse approximation factor, i.e., the short vector in $\lat^*$ does not need to be quite as short as our final output vector. We can therefore make the recursive call on the dual lattice $\lat^*$ with a lower depth parameter, and avoid the infinite loop.

To finish specifying the reduction, we need to define two base cases. In one base case, the input lattice has rank $n = k$, in which case the reduction simply uses its $\sqrt{\delta_k}$-HSVP oracle on rank-$k$ lattices to output a short non-zero lattice vector in $\lat$. In the other base case, the depth parameter $\tau$ is zero, in which case the algorithm uses an efficient procedure such as the LLL algorithm to output a not-too-long vector in the lattice.

Here is pseudocode for the reduction $\A(\lat, \tau)$ described above.

\begin{enumerate}
    \item {\bf Base cases:}
    \begin{enumerate}
        \item If $\tau = 0$, run LLL on $\lat$ and output the resulting vector.
    \item If $\rank(\lat) = k$, use an oracle for $\sqrt{\delta_k}$-HSVP to output a short non-zero lattice vector in $\lat$.
    \end{enumerate}
    \item Compute $\vec{w} \leftarrow \A(\lat^*, \tau-1)$ and output $\vec{y} \leftarrow \A(\lat \cap \vec{w}^\perp,\tau)$.
\end{enumerate}

Notice how simple this reduction is! It can be described in three short lines of pseudocode, and it makes no mention of a basis or the Gram-Schmidt orthogonalization or any other details about how one should represent the lattice $\lat$. Of course, an actual implementation of this algorithm would need some way to represent the lattice $\lat$ and some way to compute a representation of $\lat' = \lat \cap \vec{w}^\perp$,  which might be best done with bases and the Gram-Schmidt orthogonalization. But, while bases and Gram-Schmidt vectors are fundamental to understanding any basis-reduction algorithm, in this new framework, we view such things as low-level implementation details. This level of abstraction allows us to specify such reductions remarkably succinctly, as above.

Analyzing this reduction is also relatively simple, and we therefore give a complete analysis in-line in this overview section. Specifically, let $\gamma(n,\tau)$ be an upper bound on the approximation factor achieved by the above reduction when the input lattice has rank $n$ and the depth parameter is $\tau$. By definition, we have 
\begin{align*}
        \|\vec{y}\|/\det(\lat)^{\frac 1n} 
        &\leq \gamma(n-1,\tau) \cdot \det(\lat')^{\frac 1{n-1}} /\det(\lat)^{\frac 1n} 
        \\
        &=  \gamma(n - 1, \tau) \cdot  \left(\| \w \| \cdot \det(\lat)\right)^{\frac 1{n-1}}/\det(\lat)^{\frac 1 n} \\
        &= \gamma(n - 1, \tau) \left(\| \w \| / \det(\lat^*)^{\frac 1 n}\right)^{\frac 1{n-1}} \\
        &\leq \gamma(n-1,\tau) \gamma(n,\tau-1)^{\frac 1{n-1}}
    \; .
\end{align*}
So, we can take $\gamma(n,\tau)$ to satisfy the recurrence
\[
    \gamma(n,\tau) \leq \gamma(n-1,\tau) \gamma(n,\tau-1)^{\frac1{n-1}}
    \; ,
\]
with base cases $\gamma(k,\tau) \leq \sqrt{\delta_k}$ and, say, $\gamma(n,0) \leq 2^n$ (by LLL).
(The fact that the $\gamma(n,\tau-1)$ term comes with such a small exponent in the recurrence relation intuitively explains why we can afford to use a lower depth parameter in the recursive call on $\lat^*$.) A simple induction argument then shows that, e.g., $\gamma(n,\tau) \leq \delta_k^{(n-1)/(2(k-1))} \cdot 2^{n^3/2^\tau}$.

This analysis is \emph{much} simpler than the similar analysis for basis reduction algorithms, and for $\tau \gtrsim 3 \log n$, it matches the approximation factor $\gamma_{\mathsf{BR}}$ achieved by basis reduction up to the (insignificant) low-order term $2^{n^3/2^\tau}$ (which is similar to the low-order terms encountered in the basis reduction algorithms, and which of course decays rapidly as we increase $\tau$). 

One can similarly analyze the number of oracle calls $T(n,\tau)$ made by the reduction via the recurrence $T(n,\tau) = T(n,\tau-1) + T(n-1,\tau)$ with base cases $T(n,0) = 0$ and $T(k,\tau) = 1$.  It is easy to check that this recurrence is solved by the binomial coefficients $T(n,\tau) = \binom{n-k+\tau-1}{\tau-1} \approx n^\tau$. In particular, for $\tau \approx 3 \log n$, we get a number of oracle calls bounded by $n^{\Theta(\log n)}$. (To bound the final running time, we must also worry about the time needed to compute things such as $\lat^*$ and $\lat \cap \vec{w}^\perp$, which depends on the representation of $\lat$ and which we will discuss in \cref{sec:bitlengths_intro}. We ignore this distinction for now and largely conflate the number of oracle calls with the running time.)

This is already quite interesting, since in the context of cryptography (in which the oracle is often instantiated with a $2^{\Theta(n)}$-time algorithm) a factor of $n^{\Theta(\log n)}$ in the running time is not particularly significant. But, it's not ideal, and we would certainly prefer to find a variant of this reduction that achieves essentially the same approximation factor with only polynomial running time.

Fortunately, this is possible, as we describe below. 

\subsection{An efficient reduction to a harder problem (and from a harder problem)}
\label{sec:intro-DSP-to-DSP}

Intuitively, the reason that the above reduction requires superpolynomial time is because it produces a rather unbalanced binary tree of recursive calls with height $n-k \approx n$ when the input lattice $\lat$ has rank $n$. We would naturally prefer a more balanced tree with height $O(\log n)$.  Of course, the reason that the depth of the tree is so large is because the sublattice $\lat \cap \vec{w}^\perp$ has rank just one less than $\lat$. I.e., we reduce the rank of the lattice by one in each step, and therefore it takes $n-k$ steps to get down to rank $k$. So, it is natural to try to reduce the rank more quickly, perhaps reducing the rank by a constant factor at every step, or more accurately, reducing by a constant factor the difference $n-k$ between the current rank $n$ and the rank $k$ in which our oracle works.

So, suppose instead of finding a single short non-zero vector $\vec{w} \in \lat^*$, we found a whole \emph{dense sublattice} $\lat' \subset \lat^*$ with rank $\ell^* \geq 1$, i.e., a sublattice $\lat'$ of the dual with relatively small determinant. Just like before, the intersection $\lat \cap (\lat')^\perp$ %
of $\lat$ with the subspace orthogonal to $\lat'$ is a sublattice of $\lat$, now with rank $n-\ell^*$. Furthermore, we have $\det(\lat \cap (\lat')^\perp) = \det(\lat') \det(\lat)$ (provided, again, that $\lat'$ is primitive, which we may assume without loss of generality; see \cref{sec:prelims-dual}). So, finding a dense sublattice with rank $n-\ell^*$ in $\lat$ is exactly equivalent to finding a dense sublattice with rank $\ell^*$ in $\lat^*$.

We therefore generalize the reduction from the previous section by switching the oracle for $\gamma$-HSVP with an oracle for the $\gamma$-Densest Sublattice Problem ($\gamma$-DSP). In this problem, the input is a basis for a lattice $\lat \subset \R^n$ and \emph{also} a rank $\ell$ with $1 \leq \ell \leq n-1$. The goal is to find a sublattice $\lat'$ of $\lat$ with rank $\ell$ with $\det(\lat') \leq \gamma \det(\lat)^{\ell /n}$. Notice that when $\ell = 1$, this is exactly $\gamma$-HSVP. However, the problem is still interesting for larger $\ell$. Indeed, $\gamma$-DSP has been studied quite a bit, e.g., in ~\cite{gamaRankinConstantBlockwise2006, DMAlgorithmsDensestSublattice2013,liApproximatingDensestSublattice2014,dadushApproximatingCoveringRadius2019,LNCompleteAnalysisBKZ2020,WalConvergenceSlidetypeReductions2021,LWImprovingConvergencePracticality2023}. 
For example, it is known that the LLL algorithm solves $\gamma$-DSP with an approximation factor of $\gamma \leq 2^{\ell (n-\ell)/4}$~\cite{PTSublatticeDeterminantsReduced2008}. (See \cref{sec:context,sec:prelims-DSP} for more discussion of DSP.)

This leads naturally to the following idea. To solve a DSP instance $(\lat, \ell)$ where $\lat$ has rank $n$, we first recursively solve the DSP instance $(\lat^*, \ell^*)$ for some $\ell^*$ (of course with a lower depth parameter), receiving as output a dense dual sublattice $\lat' \subset \lat$ with rank $\ell^*$. We then recursively solve the DSP instance $(\lat \cap (\lat')^\perp,\ell)$, where the lattice $\lat \cap (\lat')^\perp$ has rank $n-\ell^*$. In particular, if we choose $\ell^*$ to be, say, $\ceil{(n-k)/20}$, then the difference $n-k$ will shrink by a factor of at least $19/20$ at every step.

To make sense of this, we will need base cases like before for when the depth parameter $\tau$ drops to zero and when the rank $n$ drops to $k$. In particular, when the depth parameter $\tau$ drops to zero, we simply run the LLL algorithm like we did before (though we now use it to find a whole dense sublattice, rather than a single short vector). 

However, we also need to handle a new corner case: if $\ell > n-\ell^*$, then the above does not make sense, since in our recursive call we will be asking for a rank-$\ell$ sublattice of a lattice with rank $n-\ell^*$, which is less than $\ell$! (Notice that, even if we start running this reduction with $\ell = 1$, some of the calls in our recursive tree will have much larger $\ell$.) To fix this, we need some way to convert DSP instances with $\ell$ nearly as large as $n$ to DSP instances with much smaller $\ell$. 

Of course, duality allows us to do this! That is, if $\ell$ is getting too large relative to $n$, then instead of directly finding a dense rank-$\ell$ sublattice of the primal lattice $\lat$, we (recursively) find a dense sublattice $\lat'$ of the dual $\lat^*$ with rank $n-\ell$ and output $\lat \cap (\lat')^\perp$. We call this a \emph{duality step}. In this way, we can always keep $\ell$ small enough. Specifically, in the below we simply apply this duality step whenever $\ell > n/2$, which is a natural choice and sufficient in this context. (In the more sophisticated reduction described in \cref{sec:DSP-to-SVP}, we are more careful about when we apply this duality step.)

When $n$ drops to $k$, we will call an oracle that solves $\gamma$-DSP on rank-$k$ lattices with $\gamma$ as small as possible. This use of a DSP oracle is a major drawback of this reduction (which we will fix below) for two reasons. First, DSP is a seemingly harder problem than HSVP, and the known algorithms that solve DSP are notably slower than the known algorithms for HSVP. Indeed, the fastest known algorithm for finding an exact densest sublattice with rank $\ell$ in $n$ dimensions runs in time $\ell^{\Theta(n \ell)}$~\cite{DMAlgorithmsDensestSublattice2013}, which in the worst case gives $n^{\Theta(n^2)}$. So, our DSP oracle calls might need to run in time $k^{\Theta(k^2)}$, compared to $2^{\Theta(k)}$ for HSVP oracle calls.  (See \cref{sec:context} for more discussion.) Second, it is not as clear how small we can take $\gamma$ to be for our $\gamma$-DSP oracle. The analogue of Hermite's constant $\delta_n$ in this setting is \emph{Rankin's constant},
\[
    \delta_{n,\ell} := \sup_{\lat \in \mathscr{L}(n)} \min_{\substack{\lat' \subseteq \lat \\ \rank(\lat') = \ell}} \frac {\det(\lat')^2}{\det(\lat)^{\frac {2\ell}{n}}}
    \; .
\]
Rankin's constant is not nearly as well understood as Hermite's constant (except for $\delta_{n,n} = 1$ and $\delta_{n,1} = \delta_{n,n-1} = \delta_n$). Instead, $\delta_{n,\ell}$ is only known up to a constant factor in the exponent, $\delta_{n,\ell} = n^{\Theta(\ell(n-\ell)/n)}$ (though the best known upper and lower bounds on this constant factor in the exponent are not so far from each other~\cite{hanrotImprovedAnalysisKannan2007}). However, a reasonable guess is that $\delta_{n,\ell} \approx n^{\ell(n-\ell)/(2n)}$. (See, e.g.,~\cite{shapiraVolumeEstimateSet2014} and the references therein.)
For now, let us therefore simply suppose that we have access to an oracle that solves $k^{\ell(k-\ell)/(2k)}$-DSP on rank-$k$ lattices.\footnote{Because of the uncertainty in Rankin's constant, we do not know whether this is even possible---i.e., there might be rank-$k$ lattices that do not even have a dense enough sublattice for this to work. So, in \cref{sec:DSP-to-DSP} we are more careful about this, but in this less formal overview we proudly forge ahead under this simplifying assumption. We will anyway later show how to replace the DSP oracle with an HSVP oracle by carefully modifying the reduction. Once we have removed DSP oracle calls, this issue with Rankin's constant will go away.} 

With all of this out of the way, we can finally present our reduction $\A(\lat, \ell, \tau)$, as below. Here, we leave $\ell^*$ unspecified, so that we actually give a family of algorithms depending on how $\ell^*$ is chosen (possibly depending on $\ell$, $n$, $k$, and $\tau$).

\begin{enumerate}
    \item {\bf Base cases:}
    \begin{enumerate}
    \item If $\ell > n/2$, output $\lat \cap \A(\lat^*, n-\ell,\tau)^\perp$, where $n := \rank(\lat)$.
    \item If $n = k$, use an oracle for $k^{\ell(k-\ell)/(2k)}$-DSP to output a dense sublattice of $\lat$ with rank $\ell$.
    \item If $\tau = 0$, run LLL on $\lat$ and output the lattice generated by the first $\ell$ vectors of the resulting basis.
    \end{enumerate}
    \item Compute $\lat' \leftarrow \A(\lat^*, \ell^*, \tau-1)$ and output $\lat'' \leftarrow \A(\lat \cap (\lat')^\perp,\ell,\tau)$.
\end{enumerate}

Notice that the reduction is still quite simple.
And, the analysis of the approximation factor $\gamma(n,\ell,\tau)$ achieved by this algorithm is quite similar to the above. In particular, a similar argument to the above shows that $\gamma(n,\ell,\tau)$ satisfies the recurrence
\begin{equation}
    \label{eq:intro_DSP_recurrence}
    \gamma(n,\ell,\tau) \leq \gamma(n-\ell^*,\ell,\tau) \gamma(n,\ell^*,\tau-1)^{\frac{\ell}{(n-\ell^*)}}
\; ,
\end{equation}
and after plugging in the base cases a simple argument shows that, e.g.,
\[
\gamma(n,\ell,\tau) \leq k^{\frac{\ell(n-\ell)}{2k}} \cdot 2^{\frac{n^2 \ell (n-\ell) }{C^\tau}}
\]
provided that, after applying duality, we maintain the invariant that $\ell + \ell^* \leq (2-C) n$ for some constant $1 < C < 2$. In particular, if we take $\tau \geq \Omega(\log n)$ and $\ell = 1$, we get essentially the same approximation factor as before, and therefore essentially the same approximation factor achieved by basis reduction. (See \cref{sec:DSP-to-DSP} for a full analysis of the algorithm with a specific choice of $\ell^*$.)

The benefit of this approach is, of course, that the running time is much better. In particular, notice that if we take $\ell^* \geq \beta \cdot (n-k)$ for some constant $0 < \beta < 1/2$, then the recursive calls made by the reduction form a binary tree with height bounded by 
\[
    h := \ceil{\log_{1/(1-\beta)}(n)} + \tau = O(\log n) + \tau
    \; ,
\]
so that the total size of the tree is $2^h \cdot \poly(n) \leq 2^\tau \cdot \poly(n)$. Taking $\tau = \Theta(\log n)$ gives a reduction that achieves essentially the same approximation factor as before but now with polynomially many oracle calls.

\subsection{Getting an efficient reduction (from either HSVP or DSP) to HSVP!}
\label{sec:intro-to-hsvp}

The reduction described in the previous section does not achieve what we are truly after, since it requires an oracle for DSP, rather than HSVP (and since, in order to get the claimed approximation factor, it relies on an unproven conjecture about Rankin's constant). However, our framework offers \emph{a lot} of flexibility that we have not yet exploited. In particular, we can be more careful about our choice of $\ell^*$, \emph{and} more importantly, we can try to be more clever about when we apply the duality step. Perhaps a particularly clever set of choices here will allow us to achieve a good approximation factor with a good running time while simultaneously managing to always have $\ell = 1$ whenever we make a recursive call with rank $k$.

To see why this might be possible, suppose that at some step in the reduction there is a recursive call with parameters $(n,\ell) = (2k-1,k-1)$. (Here, we are ignoring the depth parameter $\tau$ for simplicity.) If we choose $\ell^* = k-1$ here, then our reduction will make two recursive calls, with parameters respectively equal to $(n-\ell^*,\ell) = (k, k-1)$ and $(n,\ell^*) = (2k-1,k-1)$. Since the latter recursive call is the same as where we started (though, of course, the depth parameter will be lower), we can safely ignore it---if we can make the reduction work from where we started, then we should be able to make it work here too. For the call with parameters $(k,k-1)$, notice that after a duality step this becomes the base case with parameters $(k,1)$, which is what we want! Once we see that we can do this with $(n,\ell) = (2k+1,k)$, we can think about clever choices of $\ell^*$ that allow us to eventually get to this point from other values of $(n,\ell)$, etc. As it happens, this is possible if $\min\{\ell, n-\ell\} \leq n-k+1$, as we show in \cref{sec:DSP-to-SVP}.

Unfortunately, things are a bit more subtle than we are suggesting here because one must still worry about keeping the approximation factor down. In particular, the recurrence satisfied by our approximation factor given in \cref{eq:intro_DSP_recurrence} is not favorable when we take $\ell \approx \ell^* \approx n/2$. In general, we need to keep the exponent $\ell / (n - \ell^*)$ from \cref{eq:intro_DSP_recurrence} low. We do this via careful choices of parameters. In particular, in \cref{sec:DSP-to-SVP} we show a concrete polynomial-time reduction from $\gamma$-HSVP on rank-$n$ lattices to $\sqrt{\delta_k}$-HSVP on rank-$k$ lattices that achieves an approximation factor of $\gamma = (1+o(1)) \cdot \delta_k^{(n-1)/(2(k-1))}$. 

In fact, we achieve an approximation factor of $\delta_{\mathsf{DSP}} := \delta_k^{\ell ( n-\ell)/(2(k-1))} \cdot 2^{1/\poly(n)}$ for DSP whenever $1 \leq \min \{\ell,n-\ell\} \leq n -k +1$ (and not just the HSVP case, which corresponds to $\ell = 1$). (In particular, when $n \geq 2k$, this works for all $\ell$.) And, notice that we do this with only an oracle for \emph{HSVP} on rank-$k$ lattices (and \emph{not} a DSP oracle)! This actually beats the current state of the art in basis reduction in some regimes. Specifically, the best known approximation factor for DSP using basis reduction with the same parameters is essentially $\delta_{\mathsf{DSP}}^{1+k(k-2)/(n(n-\ell))}$~\cite{LNCompleteAnalysisBKZ2020}, which matches our result asymptotically for $k \ll n$ but is looser when $k = \Omega(n)$ (which is the regime most relevant for cryptography).

\subsection{On using automated search to find such reductions} 

In fact, for any fixed rank $n$, block size $k$, and number of oracle calls $C$, we can simply compute the ``best possible'' reduction of this form in the following sense. We can find a reduction that in $C$ oracle calls provably achieves an approximation factor $\gamma(n,\ell,C)$ satisfying
\begin{equation}
    \label{eq:crazy_recurrence_intro}
    \gamma(n,\ell,C) \leq \min\Big\{ \gamma(n,n-\ell,C),\ \min_{1 \leq \ell^* \leq n-k}\ \min_{0 \leq C^* < C} \gamma(n-\ell^*,\ell,C-C^*) \gamma(n,\ell^*,C^*)^{\frac{\ell}{n-\ell^*}} \Big\}
    \; ,
\end{equation}
with base cases given by $\gamma(k,1,C) = \sqrt{\delta_k}$ for $C \geq 1$ and $\gamma(n,\ell,C) \leq 2^{\ell (n-\ell)}$ if \emph{either} $C = 0$ or $n= k $ and $\ell \notin  \{1,n-1\}$ (i.e., the algorithm applies LLL if it runs out of oracle calls \emph{or} it reaches rank $k$ but requires a dense sublattice with rank $\ell > 1$ rather than simply a short vector). \cref{eq:crazy_recurrence_intro} formally captures the idea that the reduction may either apply a duality step (which explains the $\gamma(n,n-\ell,C)$ term) \emph{or} choose the best possible rank $\ell^*$ of a dense dual sublattice to find and the best possible number of oracle calls $C^*$ to allocate to finding such a dense dual sublattice.

Indeed, using a simple dynamic programming algorithm, one can compute $\gamma(n,\ell,C)$ in time  $\widetilde{O}(n^3 C^2)$. For reasonable parameters, this running time is far far less than the time required to instantiate the $\sqrt{\delta_k}$-HSVP oracle, so that precomputing these numbers $\gamma(n,\ell,C)$ can be done essentially for free. By restricting our algorithm to a sparse subset of possible choices for $C^*$, we can alternatively compute very good upper bounds on $\gamma(n,\ell,C)$ (together with a reduction witnessing this upper bound) in time $n^3 \cdot \poly(\log n, \log C)$, which is quite practical even for large numbers.
See \cref{sec:computers}.

From this perspective, the specific reduction that we present and analyze in \cref{sec:DSP-to-SVP} is perhaps best viewed as a proof of a specific upper bound on the true optimal value of $\gamma(n,\ell,C)$. Indeed, we discovered the specific approach used in \cref{sec:DSP-to-SVP} by first studying the reductions returned by an automated procedure!

See \cref{sec:computers} for discussions of how to further generalize this. E.g., there we study reductions that work with ``variable block size $k$'' by explicitly modeling the running time of a rank-$k$ oracle.

\subsection{Representing the lattices in our reductions}
\label{sec:bitlengths_intro}

We have been describing our reductions as operating on abstract lattices, 
independent of the choice of representation. 
This perspective is clean and allows remarkably succinct descriptions,
but it does not allow bounding the concrete running time of our reductions.
Indeed, choosing a representation that makes our DSP to DSP and DSP to HSVP
reductions polynomial time is not trivial. Straightforwardly
representing the lattices in our reductions with bases indeed allows
the necessary duality and intersection operations to be performed 
in time polynomial in the representation size. However, for a polynomial-time
reduction, the representation size of all lattices encountered in each recursive call of the algorithm must also remain polynomially bounded.

In \cref{section:representation}, we present two different ways to represent lattices in our framework.
First, we consider directly using (LLL-reduced) bases to represent lattices,
and show a quasi-polynomial $n^{O(\log n)}$ bound on the representation size. This yields reductions that run in time that is quasipolynomial in the input length.
To achieve polynomial representation size and thus polynomial running time, we use an approximate representation.
That is, we compute duality and intersection operations only approximately, 
in a way that keeps the representation size polynomially bounded.
The approximation is good enough for our purposes in a strong sense---namely, any solution to $\gamma$-DSP for the approximate lattice
can be lifted to a $\gamma'$-DSP solution for the original lattice,
where $\gamma' = \gamma \cdot (1 + 1/2^{\poly(n)})$ can be made arbitrarily close to $\gamma$.

	(Throughout this work, we assume that we only work with lattices $\lat \subset \Q^d$ that consist of rational vectors, to avoid thorny issues about representing real numbers. But, our reductions do make sense for lattices $\lat \subset \R^d$---provided that one represents the real numbers themselves in a sufficiently nice way. In fact, our rounded representation in \cref{section:bitlength} can be used to approximate real-valued lattices by integer lattices in a way that plays nicely with our reductions.)

\subsection{Some more context, and the relationship with prior work}
\label{sec:context}

Of course, our recursive lattice reduction framework shares much in common with basis reduction. At a very high level, both paradigms work to find a short non-zero vector in a lattice $\lat$ by first finding a dense lower-rank sublattice $\lat' \subset \lat$ and then finding a short non-zero vector in $\lat'$. So, one could argue that recursive lattice reduction is essentially a repackaging of basis reduction (though we do not know of a direct way to view our framework in terms of basis reduction or a direct way to view basis reduction in terms of our framework---see \cref{sec:THEFUTURE}). We instead think of it as a closely related but novel framework, which we hope will lead to a better understanding of lattice problems (and therefore to a better understanding of the security of lattice-based cryptography).

In this section, we simply describe some of the basis reduction algorithms that use particularly similar ideas to those in this work. (Of course, a full survey of the extensive literature on basis reduction is far beyond the scope of this paper.) As we discuss in \cref{sec:THEFUTURE}, we expect that additional ideas used in basis reduction should be useful in our new framework as well.

The technique of moving between the primal and the dual in order to find short vectors is well known in the basis reduction literature. Indeed, the basis reduction literature now contains many examples of algorithms that, like us, move between the primal and the dual in order to convert between low-rank dense sublattices of the dual and high-rank dense sublattices of the primal via the identity 
\[\det(\lat \cap (\lat')^\perp) = \det(\lat) \det(\lat')
\] for a (primitive) dual sublattice $\lat' \subset \lat^*$. In particular, the special case when $\lat'$ is generated by a single vector plays a prominent role in both Gama and Nguyen’s celebrated slide reduction algorithm~\cite{gamaFindingShortLattice2008} and Micciancio and Walter’s beautiful self-dual BKZ algorithm~\cite{micciancioPracticalPredictableLattice2016}. 
Indeed, both of these algorithms repeatedly use this identity in order to use an oracle for HSVP on rank-$k$ lattices to find a dense rank-$(k-1)$ sublattice in some rank-$k$ lattice (where the rank-$k$ lattice is itself some projection of a sublattice of the input lattice $\lat$). (Of course, in both cases, the algorithms are described in terms of the behavior of the Gram-Schmidt vectors, and not directly in terms of (projections of) sublattices. One can also argue that essentially \emph{all} basis reduction algorithms implicitly use duality whenever they use projection, via the identity $(\lat \cap (\lat')^\perp)^* = \Pi_{(\lat')^\perp}(\lat^*)$.) 
The more general technique in which $\lat'$ can have higher rank has also been used in basis reduction algorithms, including algorithms that explicitly work to find dense sublattices of a rank-$n$ lattice using an oracle to find dense sublattices of lower-rank lattices~\cite{liApproximatingDensestSublattice2014,LWImprovingConvergencePracticality2023}, as well as algorithms that use an oracle for finding dense sublattices of low-rank lattices in order to eventually find a short lattice vector in a rank-$n$ lattice~\cite{WalConvergenceSlidetypeReductions2021,LWImprovingConvergencePracticality2023}.

DSP, the computational problem of finding dense sublattices of a given rank $\ell$, has been the subject of a number of works in the basis reduction literature. E.g., it was first studied in \cite{gamaRankinConstantBlockwise2006}. And, Li and Nguyen showed how to generalize Gama and Nguyen’s slide reduction algorithm~\cite{liApproximatingDensestSublattice2014} to reduce $\gamma$-DSP with parameters $n$ and $\ell \leq k$ to $\sqrt{\delta_{k,\ell}}$-DSP with parameters $k$ and $\ell$, where $\gamma \approx \delta_{k,\ell}^{(n-\ell)/(2(k-\ell))}$. This is essentially the same high-level result that we achieve in our DSP-to-DSP reduction (described at a high level in \cref{sec:intro-DSP-to-DSP} and in detail in \cref{sec:DSP-to-DSP}), except that (1) Li and Nguyen require that $\ell \leq k$ and that $n \bmod k$ is divisible by $\ell$; (2) they only require a DSP oracle that finds rank-$\ell$ sublattices of rank-$k$ lattices, while we require one that finds rank-$\ell'$ sublattices for all $1 \leq \ell' < k$; and (3) their approximation factor also only depends on $\delta_{k,\ell}$, while ours depends on $\delta_{k,\ell'}$ for all $\ell'$ (unsurprisingly, given (2)). Walter~\cite{WalConvergenceSlidetypeReductions2021} and Li and Walter~\cite{LWImprovingConvergencePracticality2023} later showed variants of the algorithm with better performance (at least heuristically). 
	In a separate work, Li and Nguyen showed how to use basis reduction to solve DSP using an (H)SVP algorithm on rank-$k$ lattices~\cite{LNCompleteAnalysisBKZ2020}. This matches the results that we obtain with our framework in some parameter regimes, though it performs slightly worse in the important regime when $k = \Omega(n)$, as we noted above.

There are also algorithms for DSP that do not really use basis reduction. Dadush and Micciancio gave an $\ell^{O(n \ell)}$-time algorithm for finding the \emph{exact} densest sublattice with rank $\ell$~\cite{DMAlgorithmsDensestSublattice2013}. And, Dadush gave an elegant $2^{O(n)}$-time algorithm that finds a sublattice $\lat' \subseteq \lat$ with the property that $\det(\lat')^{1/\rank(\lat')}$ is within a polylogarithmic factor of the minimum possible~\cite{dadushApproximatingCoveringRadius2019}. (Notice that this algorithm does not allow us to choose the rank $\ell$ of the resulting sublattice. E.g, for some input lattices $\lat$, the algorithm might simply return $\lat$ itself, or a rank-one sublattice, or anything in between. In particular, it is not clear how to make use of this algorithm as an oracle in either our framework or in basis reduction.)

We also note that recursive approaches for basis reduction algorithms have certainly been considered in the literature~\cite{NS16,KEF21,RHFastPracticalLattice2023}. One might reasonably guess that our ``recursive lattice reduction'' framework is quite similar to these ``lattice basis reduction algorithms that use recursion,'' but in fact the use of recursion in these algorithms is fundamentally different from that in our framework. The algorithms from prior work are still largely iterative, and in particular still aim to iteratively improve the quality of a basis at each step. Recursion in these works is used primarily as a tool for minimizing the precision necessary in the Gram-Schmidt computations. Indeed, the resulting pattern of oracle calls is quite similar to the pattern in purely iterative basis reduction algorithms. Furthermore, all of these works focus on block size $k = 2$, while we are interested in larger block size (and we currently do not know of any particularly interesting instantiation of our framework with block size $k = 2$; see \cref{sec:THEFUTURE}). (Hermite also described a beautiful pair of algorithms for $2^{O(n)}$-HSVP, which are recursive and can be viewed as very early precursors of the LLL algorithm, though they are not known to terminate in polynomial time~\cite{HerExtraitsLettresCh1850}. See~\cite{NguHermiteConstantLattice2010}.)

Finally, we note that~\cite{PSRecursiveLatticeReduction2010} developed an algorithmic technique that they also call ``recursive lattice reduction.'' However, their techniques are very different from our own. Specifically, they first choose a nested sequence of sublattices $\lat_1 \subset \lat_2 \subset \cdots \subset \lat_n  = \lat$ with increasing rank and then iteratively compute a reduced basis for $\lat_i$ by taking a reduced basis for $\lat_{i-1}$, appending an appropriate vector to it, and then reducing the result. They then use this idea to break certain lattice challenges.

\subsection{On SVP vs.\ HSVP}
\label{sec:HSVP}

The reductions that we describe solve HSVP, i.e., the problem of finding a short non-zero vector in a rank-$n$ lattice $\lat$ \emph{relative to the (normalized) determinant} $\det(\lat)^{1/n}$. One can also consider the \emph{$\gamma$-approximate Shortest Vector Problem} ($\gamma$-SVP), in which the goal is to find a non-zero lattice vector $\vec{y} \in \lat_{\neq \vec0}$ such that $\|\vec{y}\| \leq \gamma \lambda_1(\lat)$,
where $\lambda_1(\lat) := \min_{\vec{x} \in \lat_{\neq \vec0}} \|\vec{x}\|$. I.e., $\gamma$-SVP asks for a non-zero lattice vector whose length is within a factor $\gamma$ of the shortest in the lattice. 

SVP is arguably the more natural problem, and is more discussed in the literature---at least in literature that is less concerned with cryptography. However, in cryptography, HSVP is typically the more important problem, and one typically judges the performance of an algorithm based on its Hermite factor, i.e., based on how well it solves HSVP. This is because in cryptography one typically encounters random lattices, which have $\lambda_1(\lat) \approx \sqrt{n/(2\pi e)} \cdot \det(\lat)^{1/n} \approx \sqrt{\delta_n} \det(\lat)^{1/n}$. (However, this is not always the case. E.g., some cryptography uses lattices with planted short vectors, or planted dense sublattices.) Indeed, in heuristic analysis of lattice algorithms, one often assumes not only that the input lattice $\lat$ satisfies $\lambda_1(\lat) \approx \sqrt{n/(2\pi e)} \cdot \det(\lat)^{1/n}$, but \emph{also} that this holds for all lattices $\lat'$ encountered by the algorithm's subroutines.

There seems to be some consensus among experts that $\gamma$-HSVP for nearly minimal $\gamma \approx \sqrt{\delta_n}$ is likely to be more-or-less as hard to solve as $\gamma'$-SVP for $\gamma'$ not much larger than one, in the sense that the fastest algorithms for $\gamma'$-SVP should be essentially as fast as the fastest algorithms for $\gamma$-HSVP for such small values of $\gamma,\gamma'$.\footnote{However, currently the fastest known algorithm with proven correctness for $\sqrt{\delta_n} \cdot \poly( \log n)$-HSVP runs in time $2^{n/2 + o(n)}$~\cite{ALSTimeAlgorithmSqrtn2021}, while the fastest known algorithm with proven correctness for $\poly(\log n)$-SVP (or even $O(1)$-SVP) runs in time $2^{0.802n}$~\cite{liuShortestLatticeVectors2011,weiFindingShortestLattice2015,aggarwalFasterSievingAlgorithm2019}. This seems to be more about our proof techniques than the inherent difficulty of the two problems, however.} But, the situation for larger approximation factors is less clear. Notice that there is a trivial reduction from $\gamma \sqrt{\delta_n}$-HSVP to $\gamma$-SVP. Furthermore, there is a (non-trivial) reduction due to Lov{\' a}sz from $\gamma^2$-SVP to $\gamma$-HSVP. So, our reductions can be composed with Lov{\' a}sz's to achieve reductions that solve $\gamma^2$-SVP (granted, at the expense of a factor of $O(n)$ in the running time and number of oracle calls, since Lov{\' a}sz's reduction requires $O(n)$ calls to a $\gamma$-HSVP oracle). However, many basis reduction algorithms that solve $\gamma_{\mathsf{BR}}$-HSVP can also be used to solve $(\gamma_{\mathsf{BR}}^2/\delta_k)$-SVP in essentially the same running time, provided that the rank-$k$ $\sqrt{\delta_k}$-HSVP oracle is replaced by a rank-$k$ exact SVP oracle (which is anyway what is used in practice).

So, while we essentially match the Hermite factor $\gamma_{\mathsf{BR}}$ of basis reduction, we do not know how to use our framework to match the approximation factor $\gamma_{\mathsf{BR}}^2/\delta_k$ that can be achieved for SVP. A similar situation holds for Micciancio and Walter's self-dual BKZ algorithm~\cite{micciancioPracticalPredictableLattice2016}, i.e., self-dual BKZ solves $\gamma_{\mathsf{BR}}$-HSVP but is not known to solve $(\gamma_{\mathsf{BR}}^2/\delta_k)$-SVP).

The same story more-or-less applies for DSP as well. In particular, we describe DSP as the problem of finding a dense sublattice that is dense \emph{relative to the determinant of the input lattice}. But, one can also consider the version of DSP in which the goal is to find a dense sublattice of rank $\ell$ whose determinant is small \emph{relative to the minimum possible}. If one wishes to distinguish between the two versions, one might use ``RDSP'' to refer to the version in which the density is relative to the determinant (where the R is in honor of Rankin).

\section{Preliminaries}

In this section, we list some preliminary definitions and results that will be necessary going forward. We note that we do not view the results here as particularly novel, though in some cases we do not know of a specific reference to prior work.

Given a lattice $\lat \subseteq \R^d$, its \emph{dual} is $\lat^* := \{\w \in \Span(\lat) : \forall \vec{y} \in \lat, \inner{\vec{w}, \vec{y}} \in \Z\} \;$.
Although this is somewhat nonstandard, it will be useful to generalize this definition to arbitrary sets $S \subseteq \R^d$ in the following natural way:
\[
    S^* := \{\w \in \Span(S) : \forall \vec{y} \in S, \inner{\vec{w}, \vec{y}} \in \Z\} \; .
\]

We say that a sublattice $\lat'$ of a lattice $\lat$ is \emph{primitive} if $\lat' = \lat \cap \Span(\lat')$. And, we call a vector $\vec{y} \in \lat$  primitive if the sublattice generated by the vector is primitive.

\subsection{Duality and sublattices}
\label{sec:prelims-dual}

Central to our results is a notion of duality for sublattices. Specifically, if $\lat'$ is a sublattice of $\lat$, then $\lat^* \cap (\lat')^\perp$ is a sublattice of the dual lattice $\lat^*$. Moreover, the \emph{duality map} $(\lat, \lat') \mapsto (\lat^*, \lat^* \cap (\lat')^\perp)$ is an involution on the set of (lattice, primitive sublattice) pairs, that preserves a natural notion of the relative density of $\lat'$ in $\lat$. In this section, we state some well known properties of the duality map. None of these results are novel, but we include proofs for completeness. (See also \cite[Section 1.3]{MarPerfectLatticesEuclidean2003} for many related properties of duality.)

\begin{lemma}[{\cite[Proposition 1.3.4]{MarPerfectLatticesEuclidean2003}}]
\label{lemma:noah}
    Suppose that $\lat'$ of rank $\ell$ is a sublattice of $\lat$ of rank $n$. Then $\Pi_{(\lat')^\perp}(\lat)$ is a lattice with rank $n - \ell$, and
    \[
        \Pi_{(\lat')^\perp}(\lat)^* = \lat^* \cap (\lat')^\perp \; . 
    \]
\end{lemma}
\begin{proof}
    We may assume without loss of generality that $\lat$ is full-rank ($d = n$). If $\lat$ is not full-rank, simply rotate $\lat$ so that it is orthogonal to the last $d - n$ coordinates and drop these coordinates; the statement to prove is invariant under rotation. 
    
    To simplify the notation, we will write $\Pi = \Pi_{(\lat')^\perp}$.
    First we will prove that $\Pi(\lat)^* = \lat^* \cap (\lat')^\perp$. Indeed, both sets are subsets of $(\lat')^\perp$. But for any $\w \in (\lat')^\perp$ and $\y \in \R^d$, $\inner{\w, \Pi(\y)} = \inner{\Pi(\w), \y} = \inner{\w, \y}$. 
    It follows that for all $\w \in (\lat')^\perp$, $\inner{\w, \Pi(\y)}$ is an integer simultaneously for all $\y \in \lat$ if and only if $\w \in \lat^*$.
    Thus 
    \[
        \Pi(\lat)^* = \Span(\Pi(\lat)) \cap \big(\lat^* \cap (\lat')^\perp\big) \;.
    \]
    But $\Span(\Pi(\lat)) = (\lat')^\perp.$ 
    Hence $\Pi(\lat)^* = \lat^* \cap (\lat')^\perp$, as claimed.

    It remains to show that $\Pi(\lat)$ is a lattice with rank $n - \ell$.  We will first show that $\Pi(\lat)^*$ has these properties! Indeed, $\Pi(\lat)^* = \lat^* \cap (\lat')^\perp$ is a lattice, since it is a subgroup of a lattice. And $\lat^* \cap (\lat')^\perp$ has rank at most $n - \ell$,  because it is contained in a $(n - \ell)$-dimensional subspace; namely, the intersection of the $n$-dimensional subspace $\Span(\lat^*)$ with the complement of its $\ell$-dimensional subset subspace $\Span(\lat')$. To see that it has rank at least $n - \ell$, notice that for every linearly independent set of lattice vectors $\vec{y}_1, \dots, \vec{y}_n \in \lat$ there exists a dual vector that has non-zero inner product with $\vec{y}_1$ and inner product zero with $\vec{y}_2, \dots, \vec{y}_n$. Thus, fixing any basis for $(\lat')$ and extending it to a set of $n$ linearly independent vectors in $\lat$, we can find $n - \ell$ dual vectors in $\lat^* \cap (\lat')^\perp$ that are all linearly independent. 

    It is clear that $\Span(\Pi(\lat))$ has dimension exactly $n - \ell$. It follows that $\Span(\Pi(\lat^*)) = \Span(\Pi(\lat))$. It is clear from this combined with the definition of dual set that $\Pi(\lat) \subseteq (\Pi(\lat)^*)^*$. But $(\Pi(\lat)^*)^*$ is a lattice, because $\Pi(\lat)^*$ is.  Thus $\Pi(\lat)$ is a subgroup of a lattice, and therefore is also a lattice.
\end{proof}

\begin{lemma}
\label{lemma:sub-basis}
    Suppose that $\lat'$ of rank $\ell$ is a primitive sublattice of $\lat$ of rank $n$. Then, for any basis $(\b_1, \dots, \b_\ell)$ of $\lat'$, there exists $\b_{\ell + 1}, \dots, \b_n$ such that $(\b_1, \dots, \b_n)$ is a basis of $\lat$.
\end{lemma}
\begin{proof}
    By \cref{lemma:noah}, $\Pi_{(\lat')^\perp}(\lat)$ is a lattice with rank $n - \ell$. Therefore, it has a basis $(\b'_1, \dots, \b'_{n - \ell})$. For $i \in [n - \ell]$, let $\b_{\ell + i} \in \lat$ be such that $\Pi_{(\lat')^\perp}(\b_{\ell + i}) = \b'_i$. 
    Suppose for contradiction that $(\b_1, \dots, \b_n)$ is not a basis for $\lat$, that is, that there is some vector $\y \in \lat$ such that $\y = a_1 \b_1 + \dots + a_n \b_n$ and not all $a_i$ are integers. Since $\Pi_{(\lat')^\perp}(\y) = a_{\ell + 1} \b'_{1} + \dots + a_n \b'_{n-\ell} \in \Pi_{(\lat')^\perp}(\lat)$, it follows that $a_{\ell + 1}, \dots, a_n$ are all integers. Thus, one of $a_1, \dots, a_\ell$ is not an integer. But then $\v := a_1 \b_1 + \dots + a_\ell \b_\ell$ is a point in $\lat \cap \Span(\lat')$ which is not an integer linear combination of $\b_1, \dots, \b_\ell$, contradicting that $\lat'$ is a primitive sublattice. 
\end{proof}
\begin{lemma}[{\cite[Corollary 1.3.5]{MarPerfectLatticesEuclidean2003}}]
\label{lemma:dsp-duality}
    Suppose that 
    $\lat'$ of rank $\ell$ is a primitive sublattice of $\lat$ of rank $n$.
    Then $\lat'' := \lat^* \cap (\lat')^\perp$ is a rank $(n - \ell)$ primitive sublattice of $\lat^*$ satisfying
    $\det(\lat'') = \det(\lat^*) \cdot \det(\lat')$.
\end{lemma}
\begin{proof}
    By \cref{lemma:noah}, $\lat'' = \Pi_{(\lat')^\perp}(\lat)^*$, where $\Pi_{(\lat')^\perp}(\lat)$ is a lattice of rank $n - \ell$. Hence $\lat''$ is also a lattice of rank $n - \ell$.  And it is a primitive sublattice of $\lat^*$ by definition, since it is the intersection of $\lat^*$ with a subspace. It remains to verify the determinant identity. Let $\b_1, \dots, \b_n$ be a basis for $\lat$ such that $\b_1, \dots, \b_\ell$ is a basis for $\lat'$, as guaranteed by \cref{lemma:sub-basis}. Then $\det(\lat) = \det((\b_1, \dots, \b_\ell)) \cdot \det((\Pi_{(\lat')^\perp}(\b_{\ell + 1}), \dots, \Pi_{(\lat')^\perp}(\b_{n}))$. The latter projected basis is evidently a basis for $\Pi_{(\lat')^\perp}(\lat)$, so this equation can be written as $\det(\lat) = \det(\lat') \cdot \det(\Pi_{(\lat')^\perp}(\lat))$. 
    Substituting in $\det(\lat) = 1 / \det(\lat^*)$ and $\det(\lat''^*) = 1 / \det(\lat'')$, we get $1/\det(\lat^*) = \det(\lat') / \det(\lat'')$, which upon rearranging yields the claimed identity. 
\end{proof}

We remark that a basis for $\lat''$ can be efficiently computed given bases for $\lat$ and $\lat'$. 

Given a rank-$\ell$ sublattice $\lat'$ of a lattice $\lat$ of rank $n$, define $\gamma(\lat, \lat') = \det(\lat') / \det(\lat)^{\ell/n}$.

\begin{lemma}
\label{lemma:duality-map}
For all $n \geq \ell \geq 1$, the duality map $(\lat, \lat') \mapsto (\lat^*, \lat^* \cap (\lat')^\perp)$ is an involution on the set of (lattice, primitive-sublattice) pairs, that preserves $\gamma$. 
\end{lemma}
\begin{proof}
Fix a lattice $\lat$, which we may again assume is full-rank without loss of generality. Let $\lat'' := \lat^* \cap (\lat')^\perp$. \cref{lemma:dsp-duality} shows that $\lat''$ is primitive. It also implies that the approximation factor is preserved:
    \begin{align*}
        \gamma = \gamma(\lat'', \lat^*) &:= \det(\lat'') / \det(\lat^*)^{\frac{n-\ell}{n}} \\
        &=\det(\lat^*) \cdot \det(\lat') / \det(\lat^*)^{\frac{n-\ell}{n}} \\
        &=\det(\lat)^{\frac{n - \ell}{n}} \cdot \det(\lat')  / \det(\lat) \\
        &=\det(\lat') / \det(\lat)^{\frac{\ell}{n}} \\
        &= \gamma(\lat', \lat).
    \end{align*}
It remains to show that duality is an involution, that is, $\lat' = \lat \cap (\lat^* \cap (\lat')^\perp)^\perp$. But since $\lat^* \cap (\lat')^\perp$ is rank $n - \ell$, $\Span(\lat^* \cap (\lat')^\perp) = \Span((\lat')^\perp)$. Thus $(\lat^* \cap (\lat')^\perp)^\perp = \Span(\lat')$, so the desired claim is simply $\lat' = \lat \cap \Span(\lat')$, which holds by the definition of primitive sublattice.
\end{proof}

\subsection{The Densest Sublattice Problem (DSP)}
\label{sec:prelims-DSP}

\begin{definition}
For $n \geq \ell \geq 1$ and $\gamma \geq 1$, the $(n, \ell, \gamma)$-approximate densest sublattice problem ($(n, \ell, \gamma)$-DSP) is defined as follows. The input is (a basis for) a rank-$n$ lattice $\lat \subset \Q^d$, for some $d \geq n$. The output is (a basis for) a sublattice $\lat' \subseteq \lat$ of rank $\ell$ satisfying $\gamma(\lat', \lat)\leq \gamma$.
\end{definition}
We refer to $\gamma(\lat', \lat)$ as the \emph{approximation factor} achieved by $\lat'$.  It is convenient to write $(n, \ell, \gamma)$-DSP($\lat$) for the set of solutions of $(n, \ell, \gamma)$-DSP on input $\lat$. 
For our results, we need two important properties of $(n, \ell, \gamma)$-DSP. First, from  \cref{lemma:duality-map}, it is immediate that $(n, \ell, \gamma)$-DSP enjoys the following self-duality property. 

\begin{corollary}
\label{corollary:duality-step}
    Let $\lat'$ be a primitive sublattice of $\lat$. Then
    $\lat' \in (n, \ell, \gamma)$-DSP($\lat)$ if and only if $\lat^* \cap \lat'^\perp \in (n, n - \ell, \gamma)$-DSP($\lat^*)$.
\end{corollary}

Second, $(n, \ell, \gamma)$-DSP behaves nicely under composition, in the following sense.
\begin{lemma}
\label{lemma:dsp-composition}
Suppose that $\lat' \in (n, m, \gamma_1)$-DSP($\lat$), and $\lat'' \in (m, \ell, \gamma_2)$-DSP($\lat'$). Then 
\[
    \lat'' \in \left(n, \ell, \gamma_2 \cdot \gamma_1^{\frac\ell m}\right)\text{-DSP}(\lat) \; .
\]
\end{lemma}
\begin{proof} We may directly calculate
    \begin{align*}
        \gamma(\lat', \lat'') \cdot \gamma(\lat, \lat')^{\frac \ell  m} &=(\det(\lat'') / \det(\lat')^{\frac \ell  m}) \cdot (\det(\lat') / \det(\lat)^{m/n})^{\frac \ell m} \\
        &=\det(\lat'')/\det(\lat)^{\frac \ell n} =: \gamma(\lat, \lat'') 
		\;.  \qedhere 
    \end{align*}
\end{proof}

\section{Reducing approximate DSP to DSP}
\label{sec:DSP-to-DSP}

We show a reduction from approximate DSP on rank-$n$ lattices to DSP on rank-$k$ lattices for $k < n$. This reduction is interesting in its own right, but we also note that the reduction and analysis serve as a useful warmup for \cref{sec:DSP-to-SVP}. On input a (basis for a) lattice $\lat$ with rank $n$, an integer $\ell$ with $n > \ell \geq 1$, and an integer $\tau \geq 0$, the reduction
$\A(\lat, \ell, \tau)$ behaves as follows.

\begin{enumerate}
\item \textbf{Duality step:} 
If $\ell > n / 2$, output $\lat \cap \A(\lat^*, n - \ell, \tau)^\perp$.

\item \textbf{Base cases: }
\begin{enumerate}
    \item \textbf{DSP:}  
    If $n = k$, use an oracle for $\gamma = \gamma(k, \ell)$-DSP to output a dense sublattice of $\lat$ with rank $\ell$.
    \item \textbf{LLL:} If $\tau = 0$, run LLL on $\lat$ and output the lattice generated by the first $\ell$ vectors of the resulting basis. 
\end{enumerate}
    \item \textbf{Recursive step:}  
    Otherwise, output $\A(\lat \cap \A(\lat, \ell^*, \tau - 1)^\perp, \ell, \tau)$, where $\ell^* = \lceil (n - k) / 20 \rceil$.
\end{enumerate}

\begin{theorem}
    \label{thm:DSP_to_DSP}
       Let $\lat$ be a lattice with rank $n \geq k \geq 10$, $\tau \geq 0$, and $1 \leq \ell < n$.
    Then on input $(\lat, \ell, \tau)$, the above reduction $\A$ runs in time $\poly(n) \cdot 2^\tau$, and returns $\lat' \subset \lat$ with rank $\ell$ with 
    \[
        \gamma' := \det{\lat'} / (\det{\lat})^{\ell / n} \leq \alpha_k^{\ell(n-\ell)} \cdot 2^{\frac{\ell(n-\ell) n^2}{2^{\tau/2}}} \; ,
    \]
    where 
    \[
        \alpha_k := \max_{1 \leq \ell' < k} \gamma(k, \ell)^{\frac 1{\ell'(k-\ell')}}
        \; .
    \]
\end{theorem}

We remark that if the DSP oracle is exact, that is, $\gamma(k, \ell) = \sqrt{\delta_{k, \ell}}$ for all $1 \leq \ell < k$, we have $\alpha_k = k^{1/k}$ and therefore $\gamma' = k^{\Theta(\ell (n-\ell)/k)}$ (where the hidden constant in the exponent is unknown but likely about $1/2$). 

\subsection{\texorpdfstring{Proof of \cref{thm:DSP_to_DSP}}{Proof of the theorem}}

In the recursive step, we call the inner recursive call a ``right child'' and the outer recursive call a ``left child,'' and we write the parameters associated with the right child as $(n_R, \ell_R, \tau_R)$ and similarly for the left child $(n_L, \ell_L, \tau_L)$.

We first observe that the reduction maintains the invariants that $n \geq k$ and and that $1 \leq \ell < n$. It is trivial to see that the reduction maintains the invariant that $1 \leq \ell < n$ by simply noting that in the recursive step the right child has $\ell_R := \ell^* =  \ceil{(n-k)/20}$  and $n_R = n \geq k$, which certainly satisfies this. The left child has $\ell_L = \ell \leq n/2$ (since if $\ell > n/2$, we would have applied a duality step and not a recursive step) and $n_L = n-\ell^* = n- \ceil{(n-k)/20} > n/2 \geq \ell_L$ and also $n_L = n-\ceil{(n-k)/20} \geq k$. Duality does not affect these invariants, so all steps maintain the invariants.

We next prove that the reduction runs in the claimed time.
Notice that the right child has $\tau_R = \tau - 1$, and the left child has $n_L - k \leq 19 (n - k)/20$. The duality step does not affect $n - k$ or $\tau$. It immediately follows that the tree of recursive steps has height $O(\log n) + \tau$. The total number of nodes in this tree is therefore bounded by $\poly(n) \cdot 2^\tau$. Furthermore, the number of duality nodes is at most the number of recursive steps, so the total number of calls to $\A$ is bounded by $\poly(n) \cdot 2^\tau$ and the running time of the reduction is similarly bounded by $\poly(n) \cdot 2^\tau$.

Finally, we bound the approximation factor. 
To that end, let $\gamma'(n,\ell,\tau)$ be the approximation factor achieved by the above reduction when called on a lattice with rank $n$ with parameters $(\ell, \tau)$. We prove by induction on $n$ and $\tau$ that 
\[
    \gamma'(n,\ell,\tau) \leq f(n,\ell,\tau) := \alpha_k^{\ell(n-\ell)} \cdot 2^{\frac{\ell(n-\ell) n^2}{2^{\tau/2}}}
    \; ,
\]
as needed. Indeed, it is trivial to check that the base cases satisfy this, since the LLL algorithm satisfies, e.g., $\gamma'(n,\ell,0) \leq 2^{\ell(n-\ell)} \leq f(n,\ell,0)$~\cite{PTUnifyingLLLInequalities2009}, and by definition the output of the DSP oracle satisfies 
$\gamma'(k,\ell,\tau) \leq \gamma(k, \ell) \leq \alpha_k^{\ell(k-\ell)} \leq f(k,\ell,\tau)$. Furthermore, the duality step does not change either $\gamma'$ or $f$, so we may ignore duality steps.

It remains to handle recursive steps. By \cref{lemma:dsp-composition}, in a recursive step, $\gamma$ satisfies the recurrence
\[
    \gamma'(n, \ell, \tau) \leq \gamma'(n - \ell^*, \ell, \tau) \cdot \gamma'(n, \ell^*, \tau-1)^{\frac{\ell }{n - \ell^*}} \;.
\] 
By induction, we have
\begin{align*}
    \gamma'(n,\ell,\tau) 
        &\leq f(n-\ell^*,\ell,\tau) f(n,\ell^*,\tau-1)^{\frac{\ell}{n-\ell^*}} \\
        &= \alpha_k^{\ell (n-\ell - \ell^*)+ \ell^* \ell} \cdot 2^{\ell (n-\ell^*-\ell)(n-\ell^*)^2/2^{\tau/2} + \ell^* \ell n^2/2^{\tau/2-1/2}}\\
        &= \alpha_k^{\ell(n-\ell)} \cdot 2^{\frac{\ell}{2^{\tau/2}}(\sqrt{2}  \ell^* n^2 + (n-\ell^*-\ell)(n-\ell^*)^2  )}
        \; .
\end{align*}
Therefore,
\begin{align*}
    \beta &:= 2^{\tau/2} \cdot \log(\gamma'(n,\ell,\tau)/f(n,\ell,\tau))\\
    &\leq \sqrt{2} \ell \ell^* n^2 + \ell (n-\ell^*-\ell)(n-\ell^*)^2 - \ell (n-\ell)n^2\\
    &= \ell \ell^* ( \ell^* (3n-\ell^*) + \ell(2n-\ell^*)- (3-\sqrt{2})n^2)
\end{align*}
Finally, recalling that $\ell \leq n/2$ and $\ell^* = \ceil{(n-k)/20} \leq n/10 $, we see that
\[
    \beta \leq \ell \ell^* \cdot (3n^2/10 + n^2 - (3-\sqrt{2})n^2) \leq 0
    \; ,
\]
i.e., $\gamma'(n,\ell,\tau) \leq f(n,\ell,\tau)$, as needed.
\qed

\section{Reducing HSVP (and even DSP!) to HSVP}
\label{sec:DSP-to-SVP}

 We now show how to modify the reduction from the previous section into a reduction from approximate DSP on rank-$n$ lattices to HSVP on rank-$k$ lattices for $k < n$. (In other words, we show how to replace the DSP oracle from the previous section with an HSVP oracle.) The high-level idea is simply to modify the above reduction in such a way that, whenever we make a recursive call with $n = k$, we will always have $\ell = 1$ or $\ell = n-1$. (Notice that $\ell = n-1$ is just as good as $\ell = 1$, since after applying duality, we have $\ell' = n-\ell = 1$. Of course, we must do this while maintaining the invariant that $\ell^*$ is typically a constant fraction of $n-k$, so that the depth of the tree will still be logarithmic in $n-k$.)
To do so, it suffices to design our recursive calls to maintain the invariant that \begin{equation}
    \label{eq:hat_ell_explanation}
\Delta :=  n-k + 1 - \min\{\ell, n-\ell\} \geq 0
\; .
\end{equation}
Notice that this is quite a natural condition to impose, given that we wish to have $\min\{\ell,n-\ell\} = 1$ when $n = k$. 

Let's refer to the inner call made by a recursive call as its \emph{left child}, with parameters $n_L, \ell_L$, $\tau_L$, and $\Delta_L := n_L - k - \min \{\ell_L, n_L-\ell_L\} + 1$; similarly, we'll refer to the outer call as the \emph{right child} and denote its parameters by $n_R, \ell_R$, and so on. (This terminology of ``right and left children'' of course comes from thinking of the recursion tree that our reduction follows.) Since right children have rank $n_R = n$ and $\ell_R = \ell^*$, maintaining \cref{eq:hat_ell_explanation} for right children simply amounts to choosing $\ell^* \leq (n-k)+1$. We actually take $\ell^* \approx (n-k)/20$ to be much smaller than this, so this is fine. 

On the other hand, left children have $n_L := n- \ell^*$ and $\ell_L = \ell$. In particular, notice that if $\ell \geq n/2$, then $\Delta_L := n_L-k+1 - (n_L - \ell)  = \Delta$, so that $\Delta_L \geq \Delta$, and we clearly maintain \eqref{eq:hat_ell_explanation} for left children. On the other hand, if $\ell < n/2$, then (ignoring the corner case when $\ell < n/2$ but $\ell_L \geq n_L/2$, which never actually occurs in our reduction), we have
\[
    \Delta_L = n_L - k + 1 - \ell = \Delta - \ell^*
    \; ,
\]
so that $\Delta_L$ is decreasing in this case, which could be bad.

To deal with this issue, we simply ``make sure that whenever we make a recursive call with $\ell \leq n/2$, $\Delta$ must be larger than $\ell^*$.'' To do that, we simply apply duality whenever $\ell \leq n/2$ and $\Delta$ is small. This amounts to applying duality whenever $\ell$ is nearly as large as $n-k$. (Notice that this is actually quite a natural way to keep $\min \{\ell, n-\ell\}$ small.)

In this way, we are able to guarantee that whenever we have $n = k$, we must have either $\ell = 1$ or $n-\ell = 1$.  (The above description ignores the fact that we must of course maintain the invariant that $\ell < n$. In order to account for that, we also apply duality if, say, $\ell \gtrsim n-(n-k)/10$.)

But, we must still of course worry about the approximation factor achieved by this approach. In particular, recall that the recurrence relation for our approximation factor in the previous section was
\begin{equation}
\label{eq:generic-gamma}
    \gamma(n,\ell,\tau) \leq \gamma(n-\ell^*,\ell,\tau) \gamma(n,\ell^*,\tau-1)^{\ell/(n-\ell^*)}
    \; .
\end{equation}
Unsurprisingly, we obtain essentially the same recurrence here. But, in the previous section, we kept the exponent $\ell/(n-\ell^*) < C$ for some constant $C < 1$, and intuitively, this meant that the contribution of the right child to the approximation factor was somehow significantly smaller than the contribution of the left child. (And, of course, this is made formal in our analysis of the recurrence.) This is what justified taking depth parameter $\tau_R = \tau-1$ for the right child.

However, with our new approach, we can no longer guarantee that this exponent ${\ell/(n-\ell^*)}$ is always small. In particular, though we would like to apply duality whenever, say, $\ell > n/2$ in order to keep $\ell$ small and therefore keep this exponent bounded, we cannot apply duality when $\ell \lesssim k$, since this could bring us back to the case $\min\{\ell,n-\ell\} \gtrsim n-k$ that we were trying to avoid above. When $n < 2k$, we could have both $\ell > n/2$ and $\ell \lesssim k$, in which case it is not clear what to do.

In this rather frustrating corner case, we simply do not lower $\tau$ for the right child (or the left child, for that matter), i.e., we take $\tau_R = \tau$ instead of $\tau_R = \tau-1$. This might seem like it would substantially increase the running time of our algorithm, or even lead to an infinite loop! Intuitively, this is not a problem because the case when $\ell > n/2$ but $\ell \lesssim k$ is relatively rare, and  indeed, we formally show that our running time is essentially unaffected by this (by showing that the grandchildren of any such node must either have lower $\tau$ or substantially smaller $n-k$).

\subsection{The algorithm}

On input a (basis for a) lattice $\lat$ of rank $n$, an integer $\ell$ with $n > \ell \geq 1$, and an integer $\tau \geq 0$, the reduction $\A(\lat, \ell, \tau)$ behaves as follows. 

\begin{enumerate}
    \item \textbf{Duality step:} If 
    $\max\{1, (n - k) / 5\} < \ell < n/2$ or $\ell \geq n - \max\{1, (n - k) / 10\}$,
    \newline output $\lat \cap \A(\lat^*, n - \ell, \tau)^\perp$. 
    \item \textbf{Base cases: }
    \begin{enumerate} 
        \item \textbf{SVP:} 
        If $n = k$ and $\ell = 1$, use an oracle for $\gamma = \gamma(k)$-HSVP to output (the lattice generated by) a short vector in $\lat$. 
        \item \textbf{LLL:} If $\tau = 0$, run LLL on $\lat$ and output the lattice generated by the first $\ell$ vectors of the resulting basis.
    \end{enumerate}
    \item \textbf{Recursive step:}
 Otherwise, set $\ell^* = \lceil (n - k) / 20 \rceil$ and return $\A(\lat \cap \A(\lat^*, \ell^*, \tau-b)^\perp, \ell, \tau)$, where $b = 1$ if $\ell < n/2$ and $b = 0$ otherwise.
\end{enumerate}

\begin{theorem}
\label{theorem:dsp-to-svp}
    Let $\lat$ be a lattice with rank $n \geq k \geq 10$, and $\tau \geq 0$, and suppose that $1 \leq \ell \leq n - k+1$. 
    Then on input $(\lat, \ell, \tau)$, $\A$ runs in time 
    $\poly(n) \cdot 4^\tau$, and returns $\lat' \subset \lat$ with rank $\ell$ satisfying 
    \[
        \gamma' := \det{\lat'} / (\det{\lat})^{\frac \ell  n} \leq \gamma(k)^{\frac{\ell(n-\ell)}{k - 1}} \cdot 2^{\frac{\ell(n-\ell) n^2}{2^\tau}} \; .
    \]
\end{theorem}

\subsection{\texorpdfstring{Proof of \cref{theorem:dsp-to-svp}}{Proof of the theorem}}
First, we must check that all calls to $\A$ have $n \geq k$ and $n > \ell \geq 1$. The initial call satisfies these conditions by assumption, so we need only check the children of a call that does satisfy these conditions. 
We can ignore the duality step, since interchanging $\ell$ with $n - \ell$ leaves the two conditions unchanged. 
The right child of a recursive case clearly satisfies the conditions, since $1 \leq \ceil{(n-k)/20} < n $ (recalling that $\ell_R = \ceil{(n-k)/20}$ and $n_R = n$). 
So, it only remains to consider the left child of a recursive case, which has  $n_L = n - \lceil (n - k) / 20 \rceil$.
Clearly $n_L \geq k$. Moreover, the duality step ensures that $\ell_L = \ell < n - \max\{1, (n - k)/10\} \leq n_L$, as needed. (I.e., if $\ell$ were larger than this, we would have applied duality. So, we never reach a recursive call with $\ell \geq n-\max\{1, (n - k)/10\}$.) 

Next, we must verify that any recursive call with $n = k$ must have $\ell = 1$ or $\ell = n - 1$. This follows as an immediate corollary of the following lemma, which formalizes the discussion about $\Delta$ above.
\begin{lemma}
    \label{lemma:ell-invariant}
    All recursive calls satisfy the \emph{invariant}
    $\min\{\ell, n - \ell\} \leq n - k+1$.
\end{lemma}
\begin{proof} 
As in \cref{eq:hat_ell_explanation} above, we define
\[
    \Delta(n,\ell) := n-k - \min\{\ell,n-\ell\} + 1
    \; ,
\]
and we note that it suffices to prove that $\Delta$ is non-negative for all recursive calls. 
     The initial call to $\A$ satisfies the invariant by assumption, so we need only check that the children of a call with non-negative $\Delta(n,\ell)$ will have $\Delta(n_L, \ell_L) \geq 0$ and $\Delta(n_R, \ell_R) \geq 0$. We can again ignore the duality step, because interchanging $\ell$ and $n - \ell$ clearly does not affect $\Delta$.
     The right children of recursion steps trivially satisfy $\Delta(n_R, \ell_R) \geq 0$, by our choice of $\ell_R = \ell^* = \ceil{(n-k)/20} \leq n-k+1 = n_R-k$+1. Furthermore, if $\ell \geq n/2$, then the left child of a recursion step has $n_L = n-\ell^*$ and $\ell_L = \ell$ and therefore must satisfy
     \[
        \Delta(n_L,\ell_L) = n_L-k - (n_L-\ell) +1 = n-\ell^*-k - (n - \ell^* - \ell) + 1 = \Delta + \ell^*
     \]
     So, if $\ell \geq n/2$ for a recursion step and the recursion step has non-negative $\Delta$, then so will its children.
     
     It remains to show that if a recursion step has $\ell < n/2$, then its left child still satisfies the invariant.
     To see this, notice that if $\ell \leq n/2$, $\Delta$ is actually quite large. Indeed, because of the duality step, we only reach a recursion step when with $\ell \leq n/2$ if we in fact have the stronger condition that $\ell \leq \max\{(n - k)/5,1\}$. This is enough to infer that the left child has non-negative $\Delta$, since 
     $n_L := n-\ceil{(n-k)/20} \leq n - \ceil{(n_L-k)/20}$ and 
     \[
        \ell_L := \ell \leq \max\{(n-k)/5,1\} \leq (n_L + (n_L-k)/20-k)/5 + 1 \leq n_L - k + 1
        \; ,
     \]
     as needed.
\end{proof}

Next we would like to argue that the total number of recursive calls is appropriately bounded.

\begin{lemma} 
    The total number of recursive calls made by the algorithm is bounded by $4^\tau \cdot \poly(n)$.
\end{lemma}
\begin{proof}
Notice that it suffices to argue that the total number of recursion steps in any path in the tree is bounded by $2 \tau + O(\log n)$. (Indeed, since every duality step is followed by a recursion step, it suffices to prove that the number of recursion steps and base cases is bounded by $4^\tau \cdot \poly(n)$. Of course, the size of a binary tree is at most $2 \cdot 2^h$, where $h$ is the height of the tree, so the total number of base cases and recursion steps is bounded by $2 \cdot 2^h$, where $h$ is the length of the longest path of recursion steps.)

Define the \emph{potential} $\Phi$ for a recursive call with parameter $n$, $\ell$, and $\tau$ as
\[
    \Phi = \tau + 20 \log(n - k + 1) \; .
\]
Notice that the potential of a child is always at most the potential of its parent. 
Thus it suffices to argue that every recursion step's grandchildren have potential $\Phi' \leq \Phi - 1$. Indeed, if this holds, then a simple induction argument shows the desired bound on the number of recursion steps in any path.

 It is easy to check that the left children of recursion steps have potential $\Phi_L \leq \Phi-1$ . Furthermore, when $\ell < n/2$, the right children of recursion steps have $\Phi_R \leq \Phi-1$ (since we lower $\tau$ by one). So, when $\ell < n/2$, both children already have smaller potential (and, since the potential never increases, the same is true of grandchildren). However, right children of a recursion step with $\ell \geq n/2$ have $\Phi_R = \Phi$. So, to finish the proof, we must show that all children of a right child of such a recursion step have potential $\Phi' \leq \Phi_R-1 = \Phi-1$. But, to see this, it suffices to notice that the right child of a recursion step with $\ell \leq n/2$ is either a base case (in which case there are simply no grandchildren to worry about) \emph{or} is itself a recursive step with has $\ell_R = \ell^*= \ceil{(n-k)/20} < n/2 = n_R/2$. But, we have already seen that nodes with $\ell < n/2$ have children with strictly smaller potential. The result follows.
\end{proof}

Finally, we bound the approximation factor. Let $\gamma'(n, \ell, \tau)$ be the (worst-case) approximation factor 
achieved by the reduction $\A(\lat, \ell, \tau)$ when its input lattice has rank $n$. It suffices to show that 
\[
    \gamma'(n, \ell, \tau) \leq f(n, \ell, \tau) :=  \gamma(k)^{\frac{\ell(n-\ell)}{k - 1}} \cdot  2^{\frac{\ell (n - \ell) n^2 }{ 2^\tau}} \;.
\]
    
The proof is by induction on $n$ and $\tau$. The base cases are easy to check using the HSVP and LLL guarantees 
$\gamma'(k, 1, \tau) \leq \gamma(k)$,
and  $\gamma'(n,\ell,0) \leq 2^{\ell (n-\ell)}$. So, we assume that the result holds for all $(n',\tau') < (n,\tau)$ (under the lexicographic order).
Notice that we can ignore duality steps, as both $\gamma$ and $f$ are unchanged if we replace $\ell$ with $n-\ell$.

By \cref{lemma:dsp-composition}, in a recursive step, $\gamma$ satisfies the recurrence
\[
    \gamma'(n, \ell, \tau) \leq \gamma'(n - \ell^*, \ell, \tau) \cdot \gamma'(n, \ell^*, \tau-b)^{\frac{\ell }{n - \ell^*}} \;.
\] 
Using the recurrence and the induction hypothesis we get
\[
    \gamma'(n, \ell, \tau) \leq \gamma(k)^{\frac{\ell (n - \ell)}{k - 1}} \cdot 2^{ \frac{\ell (n - \ell^* - \ell) (n - \ell^*)^2}{2^\tau} + \frac{\ell \ell^* n^2} { 2^{\tau-b}}} \; .
\]
Therefore, we have
\begin{align*}
    \alpha &:= 2^\tau \cdot \log(\gamma'(n,\ell,\tau)/f(n,\ell,\tau)) \\
        &\leq \ell(n-\ell^*-\ell)(n-\ell^*)^2 + 2^b\cdot \ell \ell^* n^2 - \ell(n-\ell) n^2\\
        &= \ell \ell^* \cdot \big(2 \ell  n + 3\ell^* n -(3-2^b) n^2- (\ell^*)^2 - \ell \ell^* \big) \; .
\end{align*}
When $b=0$, we therefore have
\[
\alpha = \ell \ell^* (2n-\ell^*)(\ell + \ell^* -n) 
\; ,
\]
which is negative since we chose $\ell^* = \ceil{(n-k)/20}$ and by the duality step we have $\ell < n-\max\{1,(n-k)/10\} \leq n-\ell^*$. On the other hand, when $b = 1$, we have $\ell \leq n/2$ by definition, but notice that by the duality step this actually implies that $\ell \leq \max\{1,(n-k)/5\} \leq n/5$ and $\ell^* = \ceil{(n-k)/20} \leq n/20+1$, so that
\[
    \alpha < \ell \ell^* n (2\ell + 3\ell^* - n) < \ell \ell^* n (2n/5 + 3n/20 + 3 -n) = \ell \ell^* n^2( -9n/20 + 3)
    \;, 
\]
which is clearly negative for $n \geq 10$. 

Therefore, $\alpha \leq 0$ in both cases, i.e., $\gamma'(n,\ell,\tau) \leq f(n,\ell,\tau)$, as needed. This completes the proof of \cref{theorem:dsp-to-svp}.

\section{Representing the lattices in our reductions}
\label{section:representation}

So far, we have described our reductions as operating on abstract lattices. To make them fully specified, and to prove concrete bounds on their running times, we must choose a 
way of representing lattices and performing operations on them.
We will deliberately avoid being completely formal by what we mean by such a \emph{representation} for the lattices in our reductions.
Our reductions involve two fundamental operations; duality, which maps $\lat$ to its dual $\lat^*$, and intersection, which maps $\lat$ to $\lat \cap \M^\perp$, where $\M$ is a relatively dense sublattice of the dual $\lat^*$. So, intuitively, a representation (1) assigns each lattice $\lat \subset \Q^d$ to a bitstring $[\lat] \in \{0, 1\}^*$; and (2) defines how to compute $[\lat^*]$ and $[\lat \cap \M^\perp]$ given $[\lat]$ and $[\M]$. 
(In fact, we are being a bit misleading here in two ways. First, the representation $[\lat]$ of $\lat$ might not be unique. Second, in our second example representation, the algorithm will need a bit more information than just $[\lat]$ and $[\M]$ in order to compute $[\lat \cap \M^\perp]$.)

We now describe two specific representations for concreteness. We will pay particular attention to the subtle question of how the representation affects the running time, and in particular whether the size of the representation blows up as we descend the recursion tree. (Of course, our reductions work with any representation, and different representations will have different benefits. We view the two that we chose below merely as proofs of concept.)

In the first, a lattice $\lat$ is represented by (any) LLL-reduced basis $\basis$ that generates it (and $\basis$ is in turn represented by the list of its numerators and denominators, written in binary). It is easy to see that LLL-reduced bases for $\lat^*$ and $\lat \cap \M^\perp$ can be computed in polynomial time from LLL-reduced bases for $\lat$ and $\M$. (See \cref{app:computing_stuff}.) Using this representation, we show that all of the lattices encountered by our reduction can be represented with bitlength bounded by $n^{O(\log n)} \cdot b_0$ where $b_0$ is the bitlength of the input lattice. This yields $n^{O(\log n)}$-time reductions using a very simple and natural representation.

The second representation is approximate and parameterized by a global bound $m'$ on the representation bitlength. In it, a rank-$n$ lattice $\lat \subset \Q^d$ with basis $\basis$ is represented by some other basis $\basis' \in \Q^{d \times n}$ (not necessarily a basis for $\lat$!) written as a list of numerators and denominators in binary, with the following two properties. First, $\basis'$ has bitlength at most $m'$. Second, defining $\lat' = \lat(\basis')$, we have that for any $\gamma \leq 2^{n^2}$ and any sublattice $\lat'' = \lat(\basis' \mathbf{Z})$ of $\lat'$ with $\det(\lat'') \leq \gamma \det(\lat')^{\ell/n}$, the corresponding sublattice $\lat''' = \lat(\basis \mathbf{Z})$ of $\lat$ satisfies $\det(\lat''') \leq \gamma' \det(\lat)^{\ell / n}$, for $\gamma' \leq (1 + 2^{\poly(n) - m'}) \cdot \gamma $. The bound $m'$ on the bitlength lets us achieve a polynomial running time, and the additional small factors in the approximation are insignificant for $m'$ a sufficiently large polynomial in $n$. (See \cref{section:bitlength} for discussion of this representation, and particularly \cref{thm:rounded_basis}.)

We elaborate below.

\subsection{Bounding the size of the LLL-reduced basis representation by \texorpdfstring{$n^{O(\log n)}$}{}}
\label{section:n-log-n}

\newcommand{\bC}{\mathbf{C}}
\newcommand{\bA}{\mathbf{A}}

We first simply consider the case in which our lattices are always represented by LLL-reduced bases. In \cref{app:computing_stuff}, we note that we can efficiently compute duals and intersections using this representation. The rest of this section is therefore devoted to bounding the growth of the bitlength of such representations as the algorithm continues.

Recall that each of our reductions can be described as a tree of lattices. The root node is the input lattice, and each path from the root to a node in the tree corresponds to a sequence of duality operations $\lat \mapsto \lat^*$ and intersection operations $\lat \mapsto \lat \cap \M^\perp$, where $\M$ is a relatively dense sublattice of the dual $\lat^*$. In particular, our approximation factor bounds always imply  
$\det(\M) \leq 2^{n^2 / 2} \det(\lat^*)^{\ell^*/n}$,
where $\ell^* := \rank(\M)$. 
We may assume without loss of generality that the input lattice is full-rank (i.e., $d = n$),
and is represented by a basis of bitlength $b_0$.
Let $\lat$ be a lattice in the reduction call tree such that the path from the root lattice to $\lat$ has $I$ intersection operations and $D$ duality operations. 
Then we claim that any LLL-reduced basis for $\lat$ has bitlength at most $(4 n)^D \cdot (b_0 + I \cdot n^2/2)$.  Using this claim, it is immediate that any LLL-reduced basis for any lattice in our HSVP to HSVP reduction with $\tau = O(\log n)$, in our DSP to DSP reduction with $\tau = O(\log n)$ and $\ell^* = \Theta(n - k)$, or in our DSP to HSVP reduction
with $\tau = O(\log n)$, has bitlength at most $n^{O(\log(n))}$, whenever the input lattice has bitlength polynomial in $n$. 
Indeed, the latter two reductions have call trees of height $O(\log n)$, which immediately implies $D, I \leq O(\log n)$, and it is easy to see that the former reduction has $D \leq \tau \leq O(\log n)$ and $I \leq n$; plugging these bounds into the claim yields the desired result.
(The above bounds on the parameters in the reductions are just those we already needed to prove bounds on the number of oracle calls each reduction makes.)

In the remainder of this subsection, we prove the claim that the bitlength of a LLL-reduced basis for a lattice encountered by one of our reductions is bounded by $(4 n)^D \cdot (b_0 + I \cdot n^2/2)$. Given a lattice $\lat \subset \Q^n$ of rank $\ell$, we define
\[
    q(\lat) := \min \{q > 0 : \lat \subseteq \Z^n / q\}
\]
and 
\[
    \Delta(\lat) := \det(\lat)^{1/\ell} \;.
\]
Finally, we define 
\[
    \beta(\lat) := \log \max \{\Delta(\lat), q(\lat)\} \;.
\]

To show that $\lat$ has bitlength at most $n^{O(\log n)}$, it suffices to show that $\beta(\lat) \leq n^{O(\log n)}$. This implication follows immediately from the following fact: any LLL-reduced basis $\basis$ for a lattice $\lat \subset \Q^n$ has bit complexity $\poly(n, \log q, \log(\Delta)) = \poly(n, \beta)$. Indeed, recall that any LLL-reduced basis for a lattice $\lat' \subset \Z^n$ has entries bounded by $2^{\poly(n)} \cdot \det(\lat')$.  And, notice that $q \basis$ is a LLL-reduced basis for $\lat' := q \lat \subset \Z^n$. Hence $q \basis$ has entries bounded by $2^{\poly(n)} \cdot \det(\lat') = 2^{\poly(n)} \cdot q^n \cdot \det(\lat)$, and $\basis$ can therefore be represented using $\poly(n, \log q, \log \det(\lat)) = \poly(n, \beta)$ bits. (See \cref{lem:LLL_bounded_entries} for a formal statement.)

To bound $\beta$, we proceed inductively. 
First, consider how intersection with $\M$ affects $\beta$.
Since $\M$ is a relatively dense sublattice of $\lat$, we have
\[
    \det(\lat') = \det(\lat) \det(\M) \leq 2^{\ell^2/2} \det(\lat)^{(n - \ell^*) / n}
\]
which implies
\[
    \Delta(\lat') \leq 2^{\ell^2/2} \cdot \Delta(\lat). 
\]
 Moreover, $q(\lat') \leq q(\lat)$ simply because $\lat' \subseteq \lat \subseteq \Z^n / q(\lat)$. It follows that
 \[\beta(\lat') 
    \leq \beta(\lat) + \ell^2/2 
    \leq \beta(\lat) + n^2/2 \; .
 \]

Next, consider duality. To bound $\Delta(\lat^*)$, notice that since $\lat \subseteq \Z^n / q$, we have\footnote{This is slightly more subtle than it might seem. Here, we are using the not-entirely-obvious fact that no sublattice of $\Z^d$ has determinant less than one. This follows from the fact that the square of the determinant of a sublattice of $\Z^d$ must be an integer.} 
\[
\det(\lat) \geq \det(\Z^n / q(\lat)) = 1/q(\lat)^n \;.
\]
Thus $\det(\lat^*) = 1/\det(\lat) \leq q(\lat)^n$, and $\Delta(\lat^*) \leq q(\lat)$. To bound $q(\lat^*)$, first recall that given a basis $\basis$ for $\lat$, the dual basis $\basis^* := \basis (\basis^T \basis)^{-1}$ is a basis for $\lat^*$. Let $\bC := q(\lat) \basis \in \Z^{n \times \ell}$, and write $\basis^* = q(\lat) \bC (\bC^T \bC)^{-1}$. Recall that the inverse of an invertible matrix $\bA \in \Z^{n \times n}$ can be written as $\vec{D} / \det(\bA)$, where the entries of $\vec{D}$ are polynomials in the entries of $\bA$ with integer coefficients. In particular, $\bA^{-1} \in \Z^{n \times n} / \det(\bA)$. It follows that 
\[
    q(\lat^*) \leq \det(\bC^T \bC) = \det(\bC)^2 = q(\lat)^{2 \ell} \det(\lat)^2 = (q(\lat) \Delta(\lat))^{2 \ell} \leq (q(\lat) \Delta(\lat))^{2 n} \;.
\]
Combining our bounds on $\Delta(\lat^*)$ and $q(\lat^*)$, we see that $\beta(\lat^*) \leq 4 n \cdot \beta(\lat)$. 

A simple inductive argument now shows that for $\lat$ as in the claim,
\[
    \beta(\lat) \leq (4 n)^D \cdot (b_0 + I \cdot n^2/2) \;,
\]
as required.

\subsection{The approximate representation}
\label{section:rep-recurrences}

\newcommand{\cM}{\mathcal{M}}

In this section, we show how to run our reductions using a fixed bound $m' \geq 10n^5$ on the bitlength of the bases that we work with throughout, with a small loss in approximation factors. To that end, in \cref{thm:rounded_basis}, we show an efficient ``rounding procedure'' that converts any basis $\basis$ for a lattice $\lat$ into a new basis $\basis'$ for a new lattice $\lat'$ such that (1) the bitlength of $\basis'$ is bounded by $m'$; and (2) for any $\gamma \leq 2^{n^2}$ and any sublattice $\lat'' = \lat(\basis' \mathbf{Z})$ of $\lat'$ with $\det(\lat'') \leq \gamma \det(\lat')^{\ell/n}$, the corresponding sublattice $\lat''' = \lat(\basis \mathbf{Z})$ of $\lat$ satisfies $\det(\lat''') \leq \gamma' \det(\lat)^{\ell / n}$, for $\gamma' \leq (1 + 2^{\poly(n) - m'}) \cdot \gamma$. 

Using this representation in our framework works as follows: each recursive call first ``rounds'' its input $\basis$ to $\basis'$, then proceeds as before to compute a basis for an appropriate sublattice $\basis' \mathbf{Z}$ of $\lat(\basis')$, and finally outputs the corresponding basis of a 
sublattice $\basis \mathbf{Z}$ of the input lattice.

The running time of the resulting reduction is clearly bounded by the size of the tree times the running time required by any individual node, which is $\poly(m')$ (plus a polynomial in the bit length $m_0$ of the original input lattice, a cost incurred only by the root node). So, we obtain polynomial-time instantiations of the DSP to DSP reduction in \cref{sec:DSP-to-DSP} and the DSP to HSVP reduction in \cref{sec:DSP-to-SVP}, using only the fact that the recursive call trees of these reduction contain only polynomially many nodes. (Similarly, we obtain another quasipolynomial $n^{O(\log n)}$-time instantiation of our HSVP to HSVP reduction described in \cref{sec:SVPtoSVP}.)

\subsubsection*{Bounding the approximation factor}

By \cref{thm:rounded_basis}, the loss in the approximation factor at each node is at most $(1 + 1/2^{n^5})$. That is, in place of \cref{eq:generic-gamma}, our generic approximation factor recurrence becomes
\[
    \gamma(n,\ell,\tau) \leq (1 + 1/2^{n^5}) \gamma(n-\ell^*,\ell,\tau) \cdot \left((1 + 1/2^{n^5}) \cdot \gamma(n,\ell^*,\tau')\right)^{\ell/(n-\ell^*)}
    \; .
\]

Noting that the exponent $\ell / (n - \ell^*)$ can be loosely upper bounded by $n$, we have
\begin{equation}
\label{eq:lossy-gamma}
    \gamma(n,\ell,\tau) \leq (1 + 1/2^{n^5})^{2n} \gamma(n-\ell^*,\ell,\tau) \cdot \gamma(n,\ell^*,\tau')^{\ell/(n-\ell^*)}
    \; .
\end{equation}

Thus, we lose one factor of at most $(1 + 1/2^{n^5})^{2n}$ for each level of the tree. Since all of our recursive call trees have depth at most $n$, we lose a factor in total of at most $(1 + 1/2^{n^5})^{2n^2}$, which is quite small.

\section{Computer-aided search for reductions, and numerical comparisons}
\label{sec:computers}

One feature of our DSP to DSP and DSP to HSVP reductions (and reductions in our framework more generally), is that the approximation factors that they achieve satisfies a simple recurrence. This means that concrete provable bounds on the approximation factor achieved by each recursive call can be computed recursively, using exactly the same pattern of recursive calls as the reduction itself. In this section, we illustrate how these concrete bounds can be used to (1) find optimal parameters for the recursion; (2) compare variants of our reductions with each other; and (3) compare our reductions to basis reduction.  For simplicity, we will restrict our attention to the case $\ell = 1$, which corresponds to reducing from HSVP. 

We stress that our purpose is to illustrate what is possible here, and certainly \emph{not} to provide a definitive analysis on the performance of our reductions (or the performance of basis reduction). In particular, we do not wish to get caught up by thorny issues about (1) the precise value of Hermite's constant in different dimensions; (2) the precise time required to solve $\sqrt{\delta_k}$-HSVP in $k$ dimensions; (3) the precise behavior of basis reduction, including the precise approximation factor achieved by LLL; etc. So, when we encounter such issues, we endeavor to make choices (detailed in \cref{sec:technical-computational-stuff}) that yield bounds that are reasonable, simple, and \emph{provably correct} (up to lower-order terms in the running time of (H)SVP algorithms, which we do not attempt to model), but certainly not optimal.\footnote{While this thicket of thorny issues is not ideal, we note that the situation is far worse in the context of basis reduction.  And, in some sense, these thorns are not the fault of the recursive lattice reduction framework. Indeed, our first thorny issue---the imprecision in Hermite's constant---of course comes with the territory (and is really rather minor). The remaining thorny issues fundamentally boil down to difficulties in nailing down the behavior of algorithms that are outside of our framework, specifically LLL and sieving.} We leave it to future work to consider such issues more carefully (mindful of the ongoing effort to understand such issues in the context of basis reduction), and here simply show what is possible when these issues are largely ignored.

With that disclaimer out of the way, we note that our framework is quite amenable to such computations. First, not only can we obtain concrete bounds for a particular reduction with particular parameter choices such as the one in \cref{theorem:dsp-to-svp}, we can use dynamic programming to compute the optimal approximation factor achievable (by a family of reductions that take roughly the same form as that in \cref{sec:DSP-to-SVP}) with a given number of oracle calls, along with a description of the reduction achieving it. And, (in stark contrast to, e.g., simulations for basis reduction) the resulting bound on the approximation factor is \emph{proven} correct (assuming that one uses proven bounds on Hermite's constant and proven bounds on the approximation factor achieved by LLL for the base cases---one can also of course plug in conjectured or heuristic bounds here).

Concretely, consider the following variant $\A(\lat, \ell, C)$ of our DSP to SVP reduction in \cref{sec:DSP-to-SVP}. The reduction is underspecified; namely, the choices of whether to apply duality, and which values $\ell^*$ and $C^*$ to use in the recursive step, are left unspecified. 
\begin{enumerate}
    \item \textbf{Duality step:}
    Possibly choose to output $\lat \cap \A(\lat^*, n - \ell, C)^\perp$ (where $n := \rank(\lat)$). Otherwise, continue.
    \item \textbf{Base cases: }
    \begin{enumerate} 
        \item \textbf{LLL:} If $C = 0$ or if $n = k$ and $\ell > 1$, run LLL on $\lat$ and output the lattice generated by the first $\ell$ vectors of the resulting basis.
        \item \textbf{SVP:} 
        If $n = k$ and $\ell = 1$, use an oracle for $\gamma = \gamma(k)$-HSVP to output (the lattice generated by) a short vector in $\lat$. 
    \end{enumerate}
    \item \textbf{Recursive step:}
 Otherwise, choose integers $1 \leq \ell^* \leq n - \max\{\ell + 1, k\}$, $0 \leq C^* < C$, and output $\A(\lat \cap \A(\lat^*, \ell^*, C^*)^\perp, \ell, C - C^*)$. 
\end{enumerate}

Notice that $C$ bounds the number of HSVP oracle calls the reduction can make, since if $C = 0$ the reduction must run LLL, and in the recursive step, the reduction must divide its budget $C$ between the two recursive calls. (And, note that the total running time of the reduction is bounded by $C \cdot \poly(n)$, assuming that the trivial operation of applying the duality step twice in a row is never chosen.)

Now we can define a quantity $\gamma'(n, \ell, C)$ that is the best provable bound (or, at least, the best bound that one can derive via the recurrences that we use throughout this paper) on the approximation factor the reduction can obtain, over all of its choices.
In detail, the base cases are $\gamma'(n, \ell, 0) := (4/3)^{\ell (n - \ell) / 4}$ and $\gamma'(k, 1, C) := \gamma(k)$, and for all other values,
\[
    \gamma'(n, \ell, C) := \min_{\hat{\ell} \in \{\ell, n - \ell\}} \; \min_{1 \leq \ell^* \leq n - \max\{\hat{\ell} + 1, k\}} \; \min_{0 \leq C^* < C} \gamma'(n - \ell^*, \hat{\ell}, C - C^*) \cdot \gamma'(n, \ell^*, C^*)^{\frac{\hat{\ell} }{n - \ell^*}} \;.
\]

The value $\gamma'(n, \ell, C)$, along with a description of the reduction achieving it (that is, the choices of $\ell^*$ and $C^*$ made by the reduction at each step) can then be straightforwardly computed by dynamic programming in space $\widetilde{O}(n^2 C)$ and time $\widetilde{O}(n^3 C^2)$. See \cref{figure:C50,figure:C100,fig:concrete-tree}.

 We will also want to consider a variant of the above reduction that can make HSVP oracle calls, not just on lattices with fixed rank $k$, but lattices of potentially different ranks. Of course, this does not make much sense unless we also ``charge our reduction more'' for oracle calls on higher-rank lattices. Concretely, suppose that $\gamma(k)$-HSVP can be solved in time $T(k)$ with rank $k$ lattices.
 We modify the reduction as follows: (1) we replace the variable $C$ bounding the number of oracle calls with a variable $T$ bounding the \emph{total running time} of oracle calls; (2) we let the condition for the LLL base case be simply $T = 0$; (3) we let the condition for the SVP base case be that $\ell = 1$ and $T \geq T(n)$; and (4) we allow the reduction to choose any $1 \leq \ell^* < n - \ell$ in the recursive case.
 The modified reduction satisfies a similar recurrence, but with corresponding changes: the SVP base case $\gamma'(n, 1, T) := \gamma(n)$ now applies for any rank $n$, provided that $T \geq T(n)$, and the bounds on $\ell^*$ in the the recursive case become
 $1 \leq \ell^* < n - \ell$.
 With these modifications, $T$ now bounds the running time of the algorithm obtained by instantiating the oracle calls made by such a reduction using an oracle that runs in time $T(n)$ (neglecting for simplicity the time taken by operations other than HSVP oracle calls).

\begin{figure}
    \centering
    \includegraphics{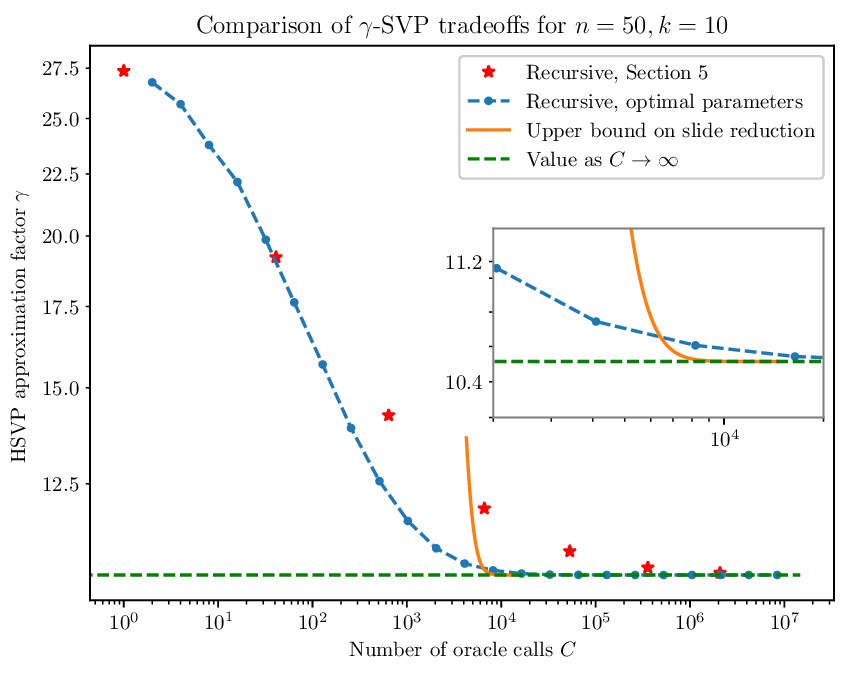}
    \caption{ 
    A comparison of tradeoffs between the approximation factor $\gamma$ and the running time (measured by the number of HSVP oracle calls $C$) when reducing $\gamma$-HSVP with rank $n = 50$ to $\sqrt{\delta_k}$-HSVP with rank $k = 10$. The blue dotted curve is our recursive DSP to HSVP reduction with nearly optimal parameters chosen by computer search. The red stars show the tradeoff achieved by the recursive reduction of \cref{theorem:dsp-to-svp}. The orange curve is an upper bound on the approximation factor obtained by slide reduction from~\cite[Corollary 1]{WalConvergenceSlidetypeReductions2021}, which we include to provide (rough) context. The blue, orange, and red curves all converge to the green dotted line as the running time grows large. }
    \label{figure:C50}    
\end{figure}

\begin{figure}
    \centering
    \includegraphics{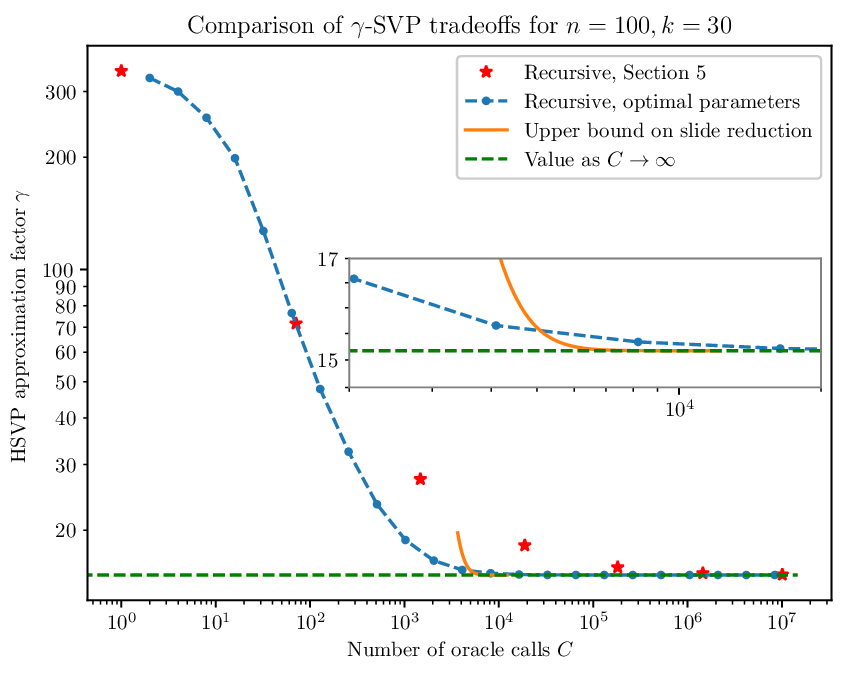}
    \caption{A similar comparison to \cref{figure:C50}, but with $n = 100$ and $k = 30$.}
    \label{figure:C100}
\end{figure}

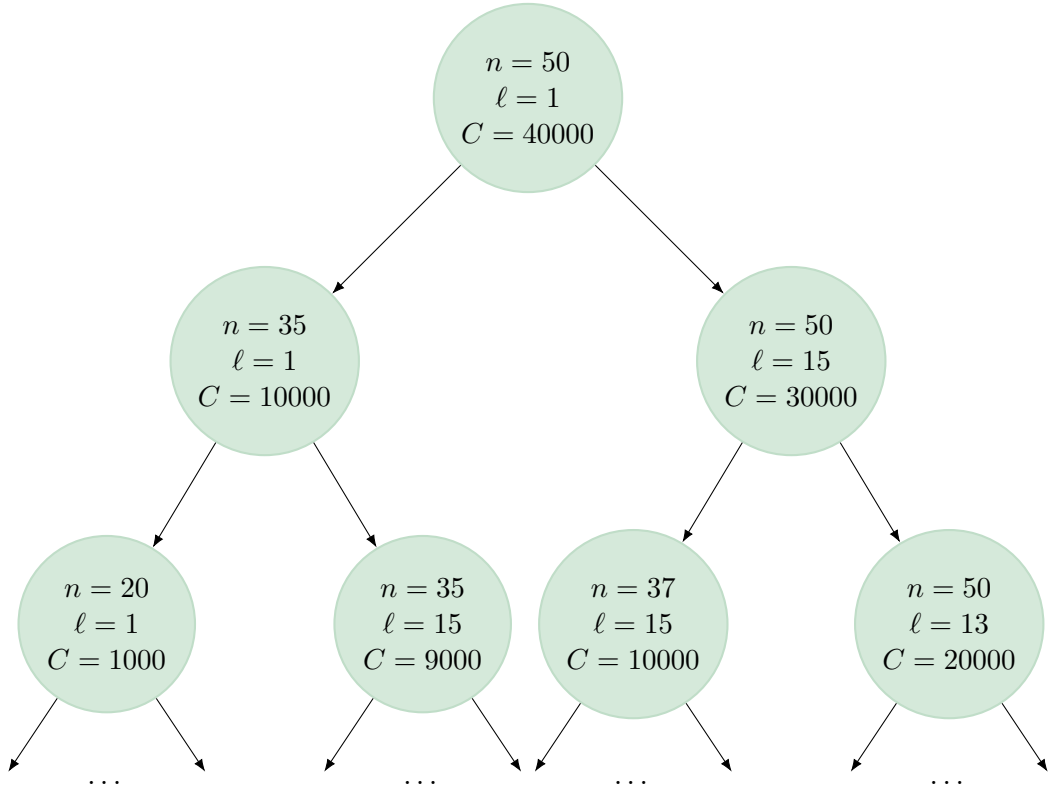
\begin{figure}
    \centering
    \begin{tikzpicture}[scale = 0.7,
        xnode/.style={circle, draw=etonblue!60, fill=etonblue!40, thick, minimum size=1.2cm, align=center},
]
\node[xnode] (A) at (0, 0) {$n = 50$ \\ $\ell = 1$ \\ $C = 40000$};
\node[xnode] (B1) at (-5, -5) {$n = 35$ \\ $\ell = 1$ \\ $C = 10000$};
\node[xnode] (B2) at (5, -5) {$n = 50$ \\ $\ell = 15$ \\ $C = 30000$};
\node[xnode] (C1) at (-8, -10) {$n = 20$ \\ $\ell = 1$ \\ $C = 1000$};
\node[xnode] (C2) at (-2, -10) {$n = 35$ \\ $\ell = 15$ \\ $C = 9000$};
\node[xnode] (C3) at (2, -10) {$n = 37$ \\ $\ell = 15$ \\ $C = 10000$};
\node[xnode] (C4) at (8, -10) {$n = 50$ \\ $\ell = 13$ \\ $C = 20000$};
\node[] (D1) at (-10, -13) {};
\node[] (D2) at (-6, -13) {};
\node[] (D3) at (-4, -13) {};
\node[] (D4) at (0, -13) {};
\node[] (D5) at (0, -13) {};
\node[] (D6) at (4, -13) {};
\node[] (D7) at (6, -13) {};
\node[] (D8) at (10, -13) {};

\foreach \from\to in {A/B1, A/B2, B1/C1, B1/C2, B2/C3, B2/C4}
\draw[-Latex] (\from) -- (\to);

\foreach \from\to in {C1/D1, C1/D2, C2/D3, C2/D4, C3/D5, C3/D6, C4/D7, C4/D8}
\draw[-Latex] (\from) -- (\to);

\foreach \from\to in {D1/D2, D3/D4, D5/D6, D7/D8}
\path (\from) -- node{\ldots} (\to);

\end{tikzpicture}
    \caption{The first few recursive calls of our recursive DSP to HSVP reduction, with optimal parameters (up to some rather aggressive coarsening) discovered by computer search. Each recursive call is labeled with the input rank $n$, output rank $\ell$, and the budget $C$ (of running time measured in HSVP oracle calls) allocated to the call. The initial parameters $n = 50, \ell = 1, C = 40000$ correspond to the regime of \cref{figure:C50} where the approximation factor achieved by the recursive reduction has gotten quite close to its value as $C$ goes to infinity (i.e., where the blue curve has more-or-less merged with the green curve).}
    \label{fig:concrete-tree}
\end{figure}

\begin{figure}
    \centering
    \includegraphics{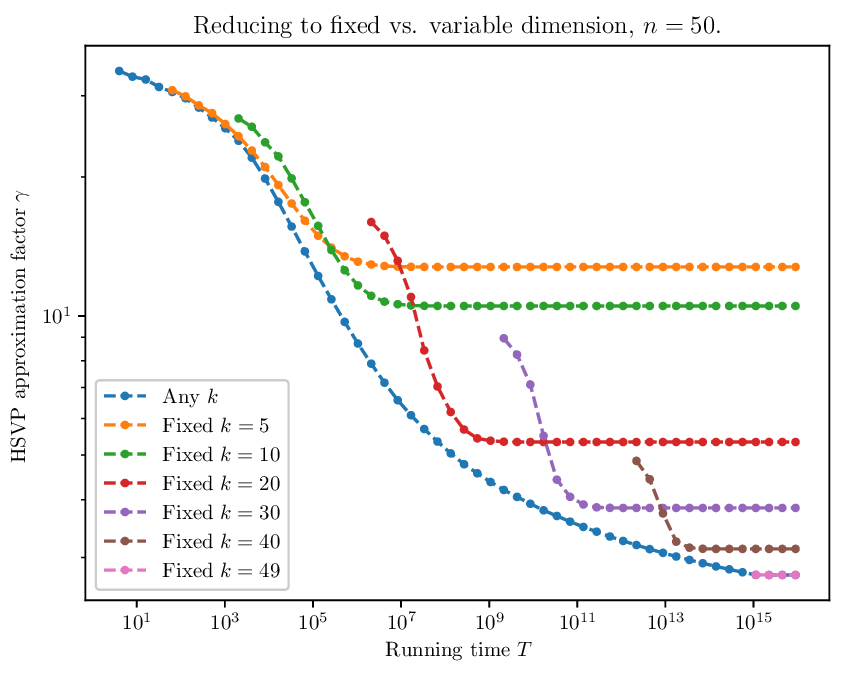}
    \caption{A comparison of tradeoffs between the approximation factor $\gamma$ and the running time $T$ for solving $\gamma$-HSVP with rank $n = 50$, in the simplified model where we assume that an HSVP oracle call with rank $k$ takes time $2^k$, and we neglect the time taken by all other operations. All curves correspond to our recursive HSVP to HSVP reduction with nearly optimal parameters chosen by computer search. The blue curve is allowed to make HSVP oracle calls with arbitrary rank $k$ (provided $2^k$ is less than its time budget $T$), whereas the others are only allowed to make HSVP oracle calls with a certain fixed rank $k$.}
    \label{figure:T50}
\end{figure}

\begin{figure}
    \centering
    \includegraphics{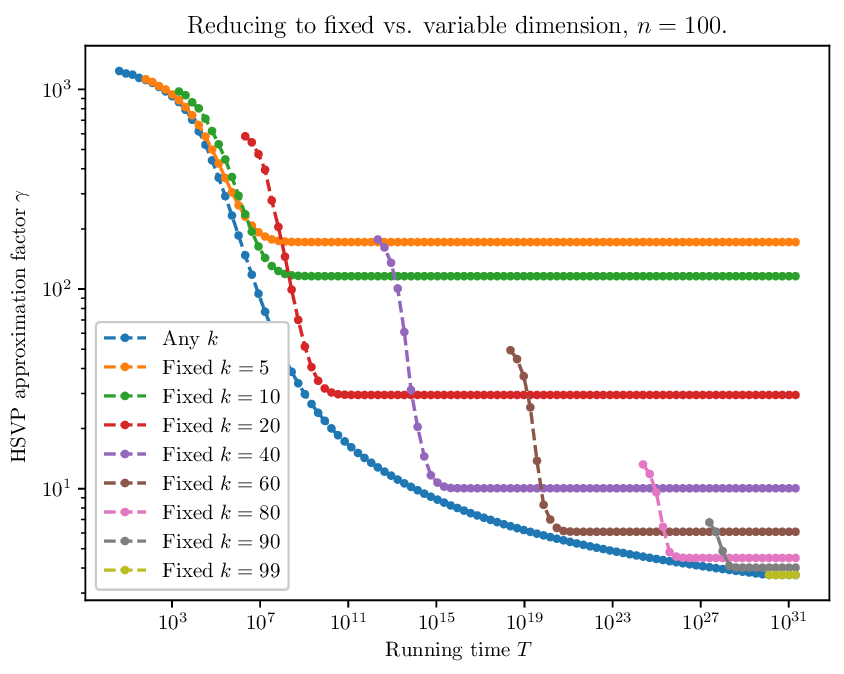}
    \caption{A similar comparison to \cref{figure:T50}, but with $n = 100$.}
    \label{figure:T100}
\end{figure}

\subsection{Results}

In \cref{figure:C50,figure:C100}, we show a comparison between the explicit reduction described in \cref{sec:DSP-to-SVP} and the optimized reduction found by a coarse-grained version of the automated procedure described above, with parameters $(n,k) = (50,10)$ and $(n,k) = (100,30)$.%
\footnote{See \cref{sec:technical-computational-stuff} for details on the coarse-grained approach.}
(The approximation factors achieved and number of oracle calls needed by the explicit reduction from \cref{sec:DSP-to-SVP} were computed by solving the associated recurrences exactly.) These plots show that the computer-generated reduction rather massively improves upon the explicit reduction described in \cref{sec:DSP-to-SVP}. 

\cref{fig:concrete-tree} shows some of the actual recursive tree in one of the reductions generated by our automated procedure. Notice that the parameters chosen here are significantly more subtle than the simple ones that we chose in \cref{sec:DSP-to-SVP}. In particular, $\ell^*$ (the value of $\ell$ taken by the right node) is clearly not just a fixed function of $n$, and the allocation of oracle calls is similarly subtle--often larger for the right node than the left.

We also include in  \cref{figure:C50,figure:C100} Walter's~\cite[Corollary 1]{walterLatticeBlogReduction} closed-form (proven) upper bound on the approximation factor achieved by (a variant of) Gama and Nguyen's slide reduction algorithm as a function of the number of oracle calls~\cite{gamaFindingShortLattice2008}. We stress that this is only meant to provide some rather rough context, to show that our (computer-generated) reduction is roughly comparable to prior work. In particular, while the plots might suggest that our reduction outperforms slide reduction when the number of oracle calls is rather small, we do not claim that this is actually true. (Slide reduction is not well studied in the setting where the number of oracle calls is significantly less than the number needed to converge. We in particular do not know if Walter's estimate is tight in the regime where the number of oracle calls is small.)

In \cref{figure:T50,figure:T100}, we show comparisons between the performance of different computer-generated reductions in our framework. In particular, we compare the time-approximation-factor tradeoff achieved by our computer-optimized reductions with different fixed oracle sizes $k$ (for $n = 50$ in \cref{figure:T50} and $n= 100$ in \cref{figure:T100}), as well as the tradeoff achieved by allowing the algorithm to adaptively choose the rank $k$ in which the HSVP oracle is run, as discussed above. As we discuss more in \cref{sec:technical-computational-stuff}, we used quite a crude cost model of $T(k) := 2^k$ computational operations for $\sqrt{\delta_k}$-HSVP oracle calls on rank-$k$ lattices, and no cost for all other operations. (We found that changing this cost model did not change the results qualitatively.)

Again, the resulting plots show that substantial improvement can be made by allowing for oracle sizes with different ranks.

\subsection{Technical details}

\paragraph{Code}
The code that we used for our computer-aided search experiments is available at \cite{code}. Running it will produce \cref{figure:C50,figure:C100,figure:T50,figure:T100}.

\label{sec:technical-computational-stuff}
\paragraph{LLL oracle}
The LLL algorithm, viewed as a $\gamma_{\mathsf{LLL}}$-DSP oracle, provably achieves approximation factor at most~\cite{PTUnifyingLLLInequalities2009}
\[
    \gamma_{\mathsf{LLL}} \leq (4/3)^{\frac{\ell (n - \ell) }{ 4}} \;.
\]
We use this bound for the approximation factor achieved by the LLL oracle.\footnote{In practice, LLL is known to perform \emph{much} better than this~\cite{NSLLLAverage2006}. So, one might instead plug in heuristic bounds on the approximation factor achieved by LLL.}

\paragraph{SVP oracle.}
We assume an exact HSVP oracle, achieving approximation factor $\gamma = \sqrt{\delta_k}$ in dimension $k$. We use known exact values of $\delta_k$ in dimensions $k \leq 8$. For $k > 8$, we use the best known explicit asymptotic upper bound on $\delta_k$, due to  Blichtfeld~\cite{BliMinimumValueQuadratic1929},
\[
    \delta_k \leq \frac 2\pi \Gamma(2 + k/2)^{\frac 2k} \;,
\]
where $\Gamma(x)$ is the gamma function.\footnote{There is a lower bound of $\Gamma(k/2+1)^{2/k}/\pi$, so Blichfeldt's bound is tight up to a factor of $2+o(1)$. One might reasonably guess that the right answer is closer to $\Gamma(k/2+1)^{\frac 2k}/\pi$.}

As a rough and simple guess for the running time $T$ required to solve (H)SVP in dimension $k$, we use $T = 2^k$.\footnote{This is approximately the running time of state-of-the-art (provable) exact SVP algorithms~\cite{aggarwalSolvingShortestVector2015}, which runs in time $2^{k+o(k)}$. But, we are not being too precise---in particular, dropping the unspecified $2^{o(k)}$ factor. One might instead plug in estimates for, e.g., heuristic sieving algorithms---or, in low dimensions, even for heuristic enumeration algorithms. Here, we only wish to give a qualitative picture of how our algorithm behaves when we allow oracle calls in different dimensions. (This qualitative picture was unchanged when we tried changing how we model the running time of the oracle.)}

\paragraph{Approximate minimization.}
Though computer-aided search for optimal recursive reductions in our framework is relatively efficient, finding the exact minimum via dynamic programming still requires non-trivial amounts of time and memory. The results presented here are therefore only approximate minima (upper bounds on the true minimum), generated by only searching for reductions where the budget $C$ of HSVP oracle calls allocated to a given recursive call cannot be an arbitrary integer, but instead must be chosen from a smaller ``coarse'' set of values. (Our approach to the running time budget $T$ is similar.) Experimentally, for the small values of $n$ for which we tested this directly, the upper bounds obtained this way were not much worse than the optima obtained by exact search, so we felt they sufficed for these proof-of-concept illustrations.

In more detail, the coarse sets of values we use are of size $O(\log C)$ (or $O(\log T)$), and they consist of all integers $\leq C$ $(\leq T)$ that can be written in base $b$ with all non-leading digits 0. (For example, choosing $b = 10$ gives $0, 1, \dots, 9, 10, 20, \dots, 90, 100, \dots$ and so on up to $C$.) We used $b = 10$ for \cref{fig:concrete-tree} (in order to have nice round numbers) and $b = 8$ for \cref{figure:C50,figure:C100,figure:T50,figure:T100} (in order to display evenly spaced points on our log-log plots). The reduction is allowed to split $C$ into $C^*$ and $C'$ with $C^* + C' = C$, which means that when $C = x \cdot b^a$, if the leading digit $x$ is 1 the reduction can choose only $C^* = y \cdot b^{a-1}$ for any $y \in \{0, \dots, b - 1\}$, and when $x > 1$, the reduction can choose only $C^* = z \cdot b^a$ for some $0 \leq z < x$.

\section{Directions for future work}
\label{sec:THEFUTURE}

We expect there to be much future work on this new paradigm. Indeed, we largely view this work as introducing new ideas that we hope and expect others will expand upon. We detail some potential future directions below.

\paragraph{New algorithms in this framework.} Perhaps the most natural direction for future work is to investigate more (hopefully better) instantiations of our new framework. For example, while we only considered binary recursion trees, one can consider trees with larger degree. One can also consider algorithms that incorporate projection more directly (in close analogy with basis reduction algorithms). E.g., one could consider algorithms that find a dense sublattice of $\lat$ with rank $\ell$ by first finding a dense sublattice $\lat'$ with rank $\ell' < \ell$, finding a dense sublattice $\lat'' \subset \Pi_{(\lat')^\perp}(\lat)$ with rank $\ell - \ell'$,  and then ``lifting'' the result to a sublattice $\widehat{\lat} \subset \lat$ with rank $\ell$ such that $\lat' \subset \widehat{\lat}$ and $\Pi_{(\lat')}(\widehat{\lat}) = \lat''$. It would be interesting both to study such approaches theoretically and to incorporate them into automated searches for better-performing reductions in this framework.

\paragraph{Recursive lattice reduction and basis reduction. } It would also be quite interesting to provide a better understanding of the relationship between our new framework of recursive lattice reduction and the pre-existing framework of basis reduction. As we discussed in \cref{sec:context}, the two frameworks feel quite similar and achieve more-or-less the same results, but we know of no formal relationship between them. In particular, we do not know whether common basis reduction algorithms can be captured in our framework, and we view this as a major open question. For example, all of our instantiations of recursive lattice reduction require the use of the LLL algorithm as a subprocedure (as we also discuss in \cref{sec:context}). Ideally, we would like to show that something like the LLL algorithm itself can be captured in our framework (or a natural generalization of it), to remove this dependency. In particular, we would at least like to match the approximation factor achieved by LLL without using LLL as a subroutine.\footnote{We do note that there is a certain weak sense in which our framework captures LLL. If one instantiates any of our algorithms with block size $k = 2$ and replaces calls to LLL with some hypothetical algorithm achieving approximation factor $2^{n^C}$ for any constant $C$, then it is not hard to see that one obtains an algorithm more-or-less matching the performance of LLL. I.e., one can use our framework to generically convert any efficient algorithm with approximation factor $2^{n^C}$ to one with approximation factor $2^{O(n)}$.} 

Failing that, one could also consider explicitly using other basis reduction algorithms as subprocedures in our framework. E.g., one could imagine including the possibility of running basis reduction algorithms instead of some of the recursive calls---in particular in automated searches for the best reductions, as in \cref{sec:computers}. Perhaps such an algorithm can achieve better performance? Or perhaps one can achieve better performance in basis reduction by using recursive lattice reduction as a subprocedure there? Perhaps there are more general combinations of the two techniques that are interesting as well.

There are also various tricks for speeding up basis reduction algorithms that we currently do not know how to apply in our framework. For example, basis reduction algorithms are often applied \emph{progressively}, in that one reduces the basis with progressively larger block size, obtaining shorter and shorter vectors~\cite{AWHTImprovedProgressiveBKZ2016}. This effectively lowers the number of high-dimensional oracle calls at the expense of using many lower-dimensional oracle calls.  While the technique we discuss in \cref{sec:computers} of using oracle calls with different dimensions is similar to this at a high level, it is not clear that they are truly analogous to progressive basis reduction algorithms. As another example, Walter~\cite{WalConvergenceSlidetypeReductions2021} and Li and Walter~\cite{LWImprovingConvergencePracticality2023} showed that it can be favorable to use the HKZ-reduced bases that one can easily generate with an SVP oracle to get a coarse solution to DSP. Perhaps similar tricks are possible here?

More generally, modern basis reduction algorithms are quite sophisticated and highly optimized---from high-level optimizations like the progressive reduction described above or optimizations that take into account the fact that the length of an output vector is not fixed but rather is a random variable~\cite{YDSecondOrderStatistical2018} to low-level (but important and very clever!)\ optimizations related to precision issues~\cite{NSLLLAlgorithmQuadratic2009,SteFloatingPointLLLTheoretical2010,RHFastPracticalLattice2023}. We hope that similarly sophisticated optimizations can and will be applied in our framework. (The fact that many of the techniques for controlling precision issues use recursion~\cite{NS16,KEF21,RHFastPracticalLattice2023} seems promising---at least superficially.)

 And, for practical performance, one should of course introduce heuristics into the study of this framework. We hope that in our setting one can reap the tremendous benefits that heuristics have provided in our understanding of basis reduction without quite needing to completely sacrifice the simple presentation and analysis that our framework provides. In particular, one can hope that it would be sufficient to, e.g., apply heuristics about Hermite's constant, the approximation factor achieved by LLL, and the other ``thorny issues'' discussed in \cref{sec:computers} without, e.g., making additional heuristic assumptions about the behavior of reductions in our framework.

\paragraph{Recursive reduction of algebraic lattices.} Additionally, one might ask how this new framework applies to algebraically structured lattices, such as embeddings of ideals in number fields or module lattices over the ring of integers. These lattices play a crucial role in lattice-based cryptography, and their reduction remains an active area of research.

It is known that the LLL algorithm~\cite{leeLLLAlgorithmModule2019} and the slide-reduction algorithm~\cite{NoahTamalika} can be readily adapted to this context. Furthermore, \cite{KEF20} have explored recursion techniques leveraging the algebraic structure of this class of lattices to enhance the efficiency of the LLL algorithm. Investigating the potential interaction between our recursive framework and the inherent recursive structure of module lattices could yield improvements to such results.

\paragraph{Lattice reduction with few oracle calls.} 
Finally, we note that our framework places new emphasis on the performance of reductions that make very few oracle calls. In particular, because of the recursive nature of our reductions, the final parameters that we obtain would be improved quite a bit if we could greatly improve the approximation factor achieved by reductions using, e.g., just one oracle call. More generally, our framework makes very concrete the intuitive notion that an improvement to the fastest running time needed to solve (H)SVP for some ranks and approximation factors can yield similar improvements for other ranks and approximation factors. This further motivates study of lattice reduction in \emph{all} parameter regimes (even rather strange regimes, like the regime of one oracle call).

\section*{Acknowledgments}

We thank the anonymous reviewers for their helpful comments. We are particularly grateful to an anonymous Crypto 2024 reviewer, who noted that the details now included in \cref{section:representation,section:bitlength} required elaboration.

NSD and SP are supported in part by the NSF under Grants Nos.~CCF-2122230 and CCF-2312296, a Packard Foundation Fellowship, and a generous gift from Google. Some of this work was done while all four authors were visiting the Simons Institute in Berkeley, as part of the summer cluster on lattices. Some of this work was done while TE, SP, and NSD were visiting DA at the National University of Singapore and the Centre for Quantum Technologies.
This work was partially supported by the grant at Centre for
Quantum Technologies titled “Quantum algorithms, complexity, and communication”.

\newcommand{\etalchar}[1]{$^{#1}$}
\def\cprime{$'$}

\appendix

\section{On the bit length of representations of lattices}
\label{section:bitlength}

In this section, we address certain technical questions that arise when one considers the number of bits required to represent the lattices $\lat$ encountered by the recursive calls to our algorithm. 

We note that much of the below is very similar to prior work---see \cref{sec:related_work_app} for a discussion on this.  To our knowledge, no prior work achieved precisely what we need here (in particular, no prior work showed a way to reduce DSP (as we parameterize it) to a variant in which the bit length is bounded), so we included the below. We stress, however, that we do not view this as an important novel contribution.

The rest of this section is organized as follows. In \cref{app:LLL_prelims}, we list some technical results about LLL-reduced bases that we will use below.  In \cref{app:LLL_bit_length}, we  bound the bit representation of an LLL-reduced basis in terms of the denominator $q$ and the determinant $\det(\lat)$ of a lattice $\lat \subset \Q^d$, which we needed in \cref{section:n-log-n}. In \cref{app:well-conditioned}, we use LLL-reduced bases and a simple rescaling technique to show that when working with DSP we can essentially assume without loss of generality that the input lattice is ``well-conditioned'' (in particular, that it has an LLL-reduced basis in which all basis vectors have roughly the same length).

In \cref{app:bounded_bits}, we show how to ``round'' any lattice to a new lattice that has a basis with bitlength bounded by $\poly(n)$ in such a way that solutions to DSP are essentially preserved. I.e., any solution to $\gamma$-DSP on the rounded lattice can be easily converted into a solution to $\gamma'$-DSP on the original lattice where $\gamma'$ can be made arbitrarily close to $\gamma$. We do this using LLL-reduced bases and the ideas from \cref{app:well-conditioned} together with a simple rounding technique.

Finally, in \cref{app:computing_stuff}, we briefly address the question of how to efficiently perform (when lattices are represented by bases) the two primitive operations performed by our algorithms: (1) computing the dual; and (2) computing $\lat \cap (\lat')^\perp$ when $\lat' \subseteq \lat^*$ is a sublattice of the dual lattice.

Note that throughout this section, we make no attempt to optimize constants.

\subsection{Some background on LLL-reduced bases}
\label{app:LLL_prelims}

Recall that given a basis $\basis = (\vec{b}_1,\ldots, \vec{b}_n) \in \Q^{d \times n}$, we define its Gram-Schmidt orthogonalization $\gs{\vec{b}}_1,\ldots, \gs{\vec{b}}_n$ by $\gs{\vec{b}}_1 := \vec{b}_1$ and 
\[
	\gs{\vec{b}}_i := \Pi_{\{\vec{b}_1,\ldots, \vec{b}_{i-1}\}^\perp}(\vec{b}_i)
	\; .
\]
We also define
\[
	\mu_{i,j} := \langle \gs{\vec{b}}_i, \vec{b}_j \rangle/\|\gs{\vec{b}}_i\|^2
	\;, 
\]
so that $\vec{b}_j = \gs{\vec{b}}_j + \sum_{i=1}^{j-1} \mu_{i,j}\gs{\vec{b}}_i$. We have the simple identity
\[
	\det(\lat) = \prod_{i=1}^n \|\gs{\vec{b}}_i\|
	\; .
\]

\begin{theorem}[{\cite[Lemma 1]{PTSublatticeDeterminantsReduced2008}}]
	\label{thm:dense_sublattice_GS}
	For any lattice $\lat \subset \Q^d$ with basis $\basis := (\vec{b}_1,\ldots, \vec{b}_n) \in \Q^{d \times n}$ and any sublattice $\lat' \subseteq \lat$ with rank $\ell$,
	\[
		\det(\lat') \geq \min_{\substack{S \subseteq [n]\\ |S| = \ell}}  \prod_{i \in S} \|\gs{\vec{b}}_i\|
	\]
\end{theorem}

We say that $\basis$ is \emph{LLL reduced} (with parameter $\delta = 3/4$) if (1) for all $i \neq j$, we have $|\mu_{i,j}| \leq 1/2$; and (2) for all $1 \leq i \leq n-1$, we have $3\|\gs{\vec{b}}_i\|^2/4 \leq  \|\gs{\vec{b}}_{i+1}\|^2 + \mu_{i,i+1}^2 \|\gs{\vec{b}}_i\|^2$. 

From the first condition, we have the simple inequality
\begin{equation}
	\label{eq:size_reduced}
	\|\vec{b}_i\| \leq \sqrt{n} \max_j \|\gs{\vec{b}}_j\|
	\; 
\end{equation}
for any LLL-reduced basis.
More interestingly, it follows immediately from the definition that LLL-reduced bases satisfy the inequality
\[
	\|\gs{\vec{b}}_i\| \leq \sqrt{2} \|\gs{\vec{b}}_{i+1}\|
	\; .
\]
By applying this inequality repeatedly, we see that for any $j \geq i$,
\begin{equation}
	\label{eq:LLL_decay}
	\|\gs{\vec{b}}_i\| \leq 2^{(j-i)/2}	 \|\gs{\vec{b}}_{j}\|
	\; .
\end{equation}

\begin{theorem}[{\cite{LLL82}}]
    \label{thm:LLL_with_bit_lengths}
    There is an algorithm that takes as input a basis $\basis \in \Q^{d \times n}$ for a lattice $\lat$ and outputs an LLL-reduced basis $\basis'$ for $\lat$ in time $O(b^5)$, where $b$ is a bound on the bitlength of the representation of $\basis$ (when written as a list of numerators and denominators in binary).
\end{theorem}

\begin{theorem}[{\cite[Theorem 1]{PTSublatticeDeterminantsReduced2008}}]
    \label{thm:LLL_orthogonality_defect}
    If $\basis = (\vec{b}_1,\ldots, \vec{b}_n) \in \Q^{d \times n}$ is LLL-reduced, then 
    \[
        \prod_{i=1}^n \|\vec{b}_i\| \leq 2^{n^2/4} \det(\lat)
    \]
\end{theorem}

Finally, we will need the following technical lemmas. 

\begin{lemma}
    \label{lem:dense_sublattice_gram_schmidts}
    If $\basis = (\vec{b}_1,\ldots, \vec{b}_n) \in \Q^{d \times n}$ is a basis for a lattice $\lat \subset \Q^d$ and $k \in [n]$, and $\alpha > \sqrt{n}$ are such that $\|\gs{\vec{b}}_i\| \leq \|\gs{\vec{b}}_j\|/\alpha$ for all $i \leq k < j$, then for any sublattice $\lat' \subseteq \lat$ with $\ell := \rank(\lat') \leq k$ and 
	\[
            \det(\lat') < (\alpha/\sqrt{n}) \cdot   \min_{|S| = \ell }\prod_{i \in S} \|\gs{\vec{b}}_i\|
	\; ,
	\]
	then we must have that $\lat' \subseteq \lat(\vec{b}_1,\ldots, \vec{b}_k)$.
\end{lemma}
\begin{proof}
    Let $\widehat{\lat} := \lat' \cap \spn(\vec{b}_1,\ldots, \vec{b}_k)$ and let $\widehat{\ell} := \rank(\widehat{\lat})$. Then, $\det(\lat') = \det(\widehat{\lat}) \det(\Pi_{\widehat{\lat}^\perp}(\lat'))$. By \cref{thm:dense_sublattice_GS}, we have that 
	\[\det(\widehat{\lat}) \geq \min_{|S| = \widehat{\ell}} \prod_{i \in S} \|\gs{\vec{b}}_i\|
	\; .
	\]
	
	Now, let $\vec{y} \in \Pi_{\widehat{\lat}^\perp}(\lat')$ be a shortest non-zero vector in $\Pi_{\widehat{\lat}^\perp}(\lat')$. Notice that by the definition of $\widehat{\lat}$, it follows that $\Pi_{\{\vec{b}_1,\ldots, \vec{b}_k\}^\perp}(\vec{y})$ is a \emph{non-zero} vector in $\Pi_{\{\vec{b}_1,\ldots, \vec{b}_k\}^\perp}(\lat)$. But, by \cref{thm:dense_sublattice_GS} again (with $\ell = 1$), any such vector must have norm at least $\min_{j > k} \|\gs{\vec{b}}_j\|$. Therefore, 
	\[
		\|\vec{y}\| \geq \|\Pi_{\{\vec{b}_1,\ldots, \vec{b}_k\}^\perp}(\vec{y})\| \geq \min_{j > k} \|\gs{\vec{b}}_j\|
		\; .
	\]
	On the other hand, by Minkowski's theorem, we have
	\[
		\det(\Pi_{\widehat{\lat}^\perp}(\lat')) \geq \|\vec{y}\|^{\ell-\widehat{\ell}}/n^{(\ell - \widehat{\ell})/2} \geq \min_{j > k} \|\gs{\vec{b}}_j\|^{\ell-\widehat{\ell}}/n^{(\ell-\widehat{\ell})/2}
		\; .
	\]

        Putting everything together, we see that
	\[
		\det(\lat') = \det(\widehat{\lat}) \det(\Pi_{\widehat{\lat}^\perp}(\lat')) \geq \min_{j > k} \|\gs{\vec{b}}_j\|^{\ell-\widehat{\ell}}/n^{(\ell-\widehat{\ell})/2} \cdot \min_{|S| = \widehat{\ell}} \prod_{i \in S} \|\gs{\vec{b}}_i\| \geq (\alpha/\sqrt{n})^{\ell - \widehat{\ell}} \cdot \min_{|S| = \ell} \prod_{i \in S} \|\gs{\vec{b}}_i\|
	\]
	But, since $\det(\lat') < (\alpha/\sqrt{n}) \cdot   \min_{|S| = \ell} \prod_{i \in S} \|\gs{\vec{b}}_i\|$ by assumption, we must have $\widehat{\ell} = \ell$. I.e., $\lat' \subseteq \lat(\vec{b}_1,\ldots, \vec{b}_k)$.
\end{proof}

\begin{corollary}
    \label{cor:other_dense_sublattice_gram_schmidt_thing}
	If $\basis = (\vec{b}_1,\ldots, \vec{b}_n) \in \Q^{d \times n}$ is a basis for a lattice $\lat \subset \Q^d$ and $k \in [n]$, and $\alpha > \sqrt{n}$ are such that $\|\gs{\vec{b}}_i\| \leq \|\gs{\vec{b}}_j\|/\alpha$ for all $i \leq k < j$, then for any primitive sublattice $\lat' \subseteq \lat$ with $\ell := \rank(\lat') \geq k$ and 
	\[\det(\lat') < \left(\frac{\alpha}{\sqrt{n \delta_{\ell,k}}} \cdot \prod_{i=1}^k \|\gs{\vec{b}}_i\|  \right)^{\ell/k}
	\; ,
	\]
	we must have $\vec{b}_1,\ldots, \vec{b}_k \in \lat'$,
	where $\delta_{\ell,k} \leq n^{n}$ is Rankin's constant (and we use the convention that $\delta_{\ell,\ell} := 1$).
\end{corollary}
\begin{proof}
    By the definition of Rankin's constant, there must be some primitive sublattice $\lat'' \subseteq \lat'$ with rank $k$ such that
	\[
		\det(\lat'') \leq \sqrt{\delta_{\ell,k}} \cdot \det(\lat')^{\ell/k} < (\alpha/\sqrt{n}) \cdot \prod_{i=1}^{k} \|\gs{\vec{b}}_i\|
				\; .
	\]
	But, by \cref{lem:dense_sublattice_gram_schmidts} it follows that $\lat'' \subseteq \lat(\vec{b}_1,\ldots, \vec{b}_k)$. And, since $\lat''$ is primitive and has rank $k$, we must have $\lat'' = \lat(\vec{b}_1,\ldots, \vec{b}_k)$. The result follows.
\end{proof}

\begin{lemma}
    \label{lem:LLL_distortion}
    If $\basis = (\vec{b}_1,\ldots, \vec{b}_n) \in \Q^{d \times n}$ is an LLL-reduced basis and $\vec{z} \in \Z^n$, then 
    \[
        \|\basis \vec{z}\| \geq \|\vec{b}_1\|\|\vec{z}\|/3^{n/2}
        \; .
    \]
\end{lemma}
\begin{proof}
    Let $\vec{z} = (z_1,\ldots, z_n)^T$. Notice that there must exist some index $i$ such that%
    \footnote{To see this, define $s_j := z_{n-j+1}^2 + \cdots + z_n^2$ with $s_{0} = 0$ and $\alpha := 2^n \|\vec{z}\|^2/(3^n-1)$. If no such $i$ exists, then for all $j$, we have 
    \[
    s_{j+1}-s_j = z_{n-j}^2 < 2^{n-j}\|\vec{z}\|^2/(3^n-1) + s_j/2 = \alpha/2^j + s_j/2
    \; 
    \]
    for all $j$. So, 
    \[
        s_{j+1} < \alpha/2^j + 3s_{j}/2
        \; .
    \]
    A simple induction argument then shows that for all $j$, 
    \[
        s_j < (3^j-1)\alpha/2^j
        \; ,
    \]
    which in particular implies that
    \[
        \|\vec{z}\|^2 = s_n < (3^n-1) \alpha/2^n = \|\vec{z}\|^2
        \; ,
    \]
    a contradiction. So, such an $i$ must exist.}
    \[z_i^2 \geq 2^{i}\|\vec{z}\|^2/(3^n-1) + \sum_{j > i} z_j^2/2
    \; .
    \] 
    Therefore, $|z_i| \geq 2^{i/2} \|\vec{z}\|/3^{n/2} + \sum_{j > i} |z_j|/2$.
    Then,
    \[
        \|\basis \vec{z}\| \geq \|\Pi_{\gs{\vec{b}}_i}(\basis \vec{z})\| = |z_i  + \mu_{i,i+1} z_{i+1} + \cdots + \mu_{i,n} z_n| \|\gs{\vec{b}}_i\| \geq (|z_i|  - (|z_{i+1}| + \cdots +  |z_n|/2))  \cdot \|\gs{\vec{b}}_i\| \geq 2^{i/2} \|\vec{z}\|\|\gs{\vec{b}}_i\|/3^{n/2}
        \; .
    \]
    Finally, recalling that the basis is LLL-reduced, it follows from \cref{eq:LLL_decay} that $\|\basis\vec{z}\| \geq \|\vec{z}\| \|\vec{b}_1\|/3^{n/2}$, as needed.
\end{proof}

\subsection{Succinct representations in terms of the denominator and the determinant}
	\label{app:LLL_bit_length}

Before we prove our main result, we go on a brief digression to observe that an LLL-reduced basis $\basis \in \Z^{d \times n}/q$ has entries whose numerators are bounded in terms of $q$ and $\det(\lat)$.

\begin{lemma}
	\label{lem:LLL_bounded_entries}
	If $\basis = (\vec{b}_1,\ldots, \vec{b}_n) \in \Z^{d \times n}/q$ for some $q \geq 2$ is LLL reduced, then $\basis \in [-p/q,p/q]^{d \times n}$ where 
	\[
		 p \leq (2^{n/4} q)^{n} \det(\lat)
	\]
	In particular, the entries in $\basis$ can be represented using $\log p + n\log q \leq (2n+1) \log q + \log \det(\lat) + 2n^2$ bits.
\end{lemma}
\begin{proof}
		Since $\basis$ is a basis, we have that $\vec{b}_i$ is non-zero for all $i$, and since $\vec{b}_i \in \Z^d/q$, we must have $\|\vec{b}_i\| \geq 1/q$. Therefore, for any $j$,
		\[
			\|\vec{b}_j\| \leq q^{n-1} \cdot \prod_{i=1}^n \|\vec{b}_i\| \leq q^{n-1} 2^{n^2/4} \det(\lat)
			\; ,
		\]	
		where the second inequality is \cref{thm:LLL_orthogonality_defect}.
		In particular, this implies that $q\vec{b}_j \in [-p,p]^{d}$, as claimed.
\end{proof}

\begin{corollary}
    \label{cor:LLL_bounded_bit_length}
    If $\lat \subset \Q^{d \times n}$ has a basis $\vec{B}$ with bit length $m$ (when represented as a list of numerators and denominators in binary), then any LLL-reduced basis of $\lat$ has bit length at most $m' \leq 4n^3m + 4n^4$.
\end{corollary}
\begin{proof}
    Let $p_{i,j}$ and $q_{i,j}$ be the numerators and denominators respectively in $\vec{B}$ so that the $(i,j)$th entry of the basis $\vec{B}$ is $p_{i,j}/q_{i,j}$. Then, $\det(\lat)^2 \leq \prod_{j=1}^n \|\vec{b}_j\|^2 \leq n^n \prod_j \max_i ( p_{i,j}/q_{i,j})^2 \leq n^n 2^{2m}$.

    By \cref{lem:LLL_bounded_entries}, we have that any LLL-reduced basis can be written using at most
    \[
        n^2 (2n+1)\log q + n^2 \log \det(\lat) + 2n^2  \leq 3n^3 \log q + n^2 m + 4n^4
        \; ,
    \]
    where $q := \prod_{i,j} q_{i,j}$. The result follows by noting that $q \leq 2^m$.
\end{proof}

\subsection{Reducing to ``well-conditioned'' bases}
\label{app:well-conditioned}

In this section, we show how to reduce DSP to the case when the input basis is ``well-conditioned'' in the sense that all of the Gram-Schmidt vectors have roughly the same length. Our main technical tool is the following. Given a basis $\basis = (\vec{b}_1,\ldots, \vec{b}_n) \in \Q^{d \times n}$ and $\gamma \geq 1$, let $A_{\basis, \gamma} \in \R^{d \times d}$ be the unique matrix such that $A_{\basis, \gamma} \gs{\vec{b}}_i = \gs{\vec{b}}_i/\alpha_i$ for all $i$, where $\alpha_1 = 1$ and 
	\[
		\alpha_i :=  \max\left\{ \alpha_{i-1}, \left\lceil \frac{\|\gs{\vec{b}}_i\|}{ \gamma^i \|\vec{b}_{1}\|} \right\rceil  \right\}
		\; ,
	\]
	and $A_{\basis, \gamma} \vec{y} = \vec{y}$ for any $\vec{y} \in \{\vec{b}_1,\ldots, \vec{b}_n\}^\perp$. (This latter condition will not be necessary for us, as we will only care about the action of $A_{\basis, \gamma}$ on $\spn(\vec{b}_1,\ldots, \vec{b}_n)$, but we include it so that $A_{\basis, \gamma}$ is well-defined.) 
	
	The following claim shows that $A_{\basis, \gamma}$  ``commutes with the Gram-Schmidt orthogonalization.''
 
 \begin{claim}
 	\label{clm:scaled_gram_schmidts}
 		For any basis $\basis = (\vec{b}_1,\ldots, \vec{b}_n) \in \Q^{d \times n}$, $\gamma \geq 1$, index $i$, and vector $\vec{y} \in \R^d$, 
 				\[
 					\Pi_{\{\vec{b}_1,\ldots, \vec{b}_i\}^\perp}(A \vec{y}) =  A\Pi_{\{\vec{b}_1,\ldots, \vec{b}_i\}^\perp}(\vec{y}) = A \Pi_{\{A\vec{b}_1,\ldots, A\vec{b}_i\}^\perp}(\vec{y})
 					\; ,
 				\]
 			where $A := A_{\basis, \gamma}$. In particular, if $\basis' := (\vec{b}_1',\ldots, \vec{b}_n') A \basis$, then
 			\[
 				\gs{\vec{b}}_i' = \gs{\vec{b}}_i/\alpha_i
 				\; ,
 			\]
 			where $\alpha_i$ is as in the definition of $A_{\basis, \gamma}$ and 
    \[
        \langle \gs{\vec{b}}_i', \vec{b}_j' \rangle/\|\gs{\vec{b}}_i'\|^2 = \langle \gs{\vec{b}}_i, \vec{b}_j \rangle/\|\gs{\vec{b}}_i\|^2
        \; .
    \]
 	 \end{claim}

    \begin{claim}
        \label{clm:Gram-Schmidts_decay_and_rescaling}
        For any basis $\basis = (\vec{b}_1,\ldots, \vec{b}_n) \in \Q^{d \times n}$ and $\gamma \geq 2$, let $\basis'= (\vec{b}_1',\ldots, \vec{b}_n')  := A_{\basis,\gamma} \basis$ and let $\alpha_i$ be as in the definition of $A_{\basis, \gamma}$. Then $\alpha_1 \leq \alpha_2 \leq \cdots \alpha_n$ where the $\alpha_i$ are as in the definition of $A_{\basis,\gamma}$, and for all $i$, if $\alpha_i < \alpha_{i+1}$, then $\|\gs{\vec{b}}_{i+1}'\| > (\gamma/2) \cdot \|\gs{\vec{b}}_i'\|$.
    \end{claim}
    \begin{proof}
        The fact that $\alpha_i$ are monotonic is immediate from the definition.
        
        Now, suppose $\alpha_i < \alpha_{i+1}$. From the definition of $\alpha_i$, we see that we must have $\|\gs{\vec{b}}_i\| < \|\gs{\vec{b}}_{i+1}^{(1)}\|/\gamma$. Then,
    \[
        \alpha_i  \geq \frac{\|\gs{\vec{b}}_i^{(1)}\|}{\gamma^i \|\vec{b}_1^{(1)}\|}
        \; ,
    \]
    and
    \[
        \alpha_{i+1} < 2 \cdot
 \frac{\|\gs{\vec{b}}_{i+1}\|}{\gamma^{i+1} \|\vec{b}_1\|}
        \; 
    \]
    (where we have used the fact that if $\ceil{x} > 1$, then $\ceil{x}/x < 2$).
    It follows that
    \[
        \|\gs{\vec{b}}_i'\| = \frac{1}{\alpha_i} \cdot \|\gs{\vec{b}}_i\| \leq \frac{\gamma^i}{\|\vec{b}_1\|} \leq \frac{2}{\gamma \alpha_{i+1}} \|\gs{\vec{b}}_{i+1}\| = \frac{2}{\gamma} \cdot \|\gs{\vec{b}}_{i+1}'\|
        \; ,
    \]
    as needed.
    \end{proof}

\begin{corollary}
    \label{cor:LLL_after_rescaling}
    If $\basis^{(1)} \in \Q^{d \times n}$ is an LLL-reduced basis and $\gamma \geq 2$, then $\basis^{(2)} := A_{\basis^{(1)}, \gamma} \basis^{(1)} = (\vec{b}_1^{(2)},\ldots, \vec{b}_n^{(2)})$ is LLL-reduced with 
    \[
    \|\vec{b}_1^{(2)}\|/2^n \leq \|\gs{\vec{b}}_i^{(2)}\| \leq \|\vec{b}_i^{(2)}\| \leq (2\gamma)^i \|\vec{b}_1^{(2)}\|
    \]
    for all $i$. 
\end{corollary}

\begin{proof}
    It follows immediately from \cref{clm:scaled_gram_schmidts} that $\mu_{i,j}^{(2)} := \langle \gs{\vec{b}}_i^{(2)}, \vec{b}_j^{(2)} \rangle/\|\gs{\vec{b}}_i^{(2)}\|^2$ satisfies $\mu_{i,j}^{(2)} = \mu_{i,j}^{(1)}$ and in particular that $|\mu_{i,j}^{(2)}| \leq 1/2$ for all $i,j$ (since $\basis^{(1)}$ is itself LLL-reduced by assumption).
    
    Now, fix some index $i$. It remains to show that $3\|\gs{\vec{b}}_i^{(2)}\|^2/4 \leq \|\gs{\vec{b}}_{i+1}^{(2)}\|^2 + (\mu_{i,i+1}^{(2)})^2 \|\gs{\vec{b}}_i^{(2)}\|^2$. Notice that since $\basis^{(1)}$ is LLL-reduced,
    \[
        3\|\gs{\vec{b}}_i^{(1)}\|^2/4 \leq \|\gs{\vec{b}}_{i+1}^{(1)}\|^2 + (\mu_{i,i+1}^{(1)})^2 \|\gs{\vec{b}}_i^{(1)}\|^2
        \; .
    \]
    Applying \cref{clm:scaled_gram_schmidts} again, we have
    \[
        3 \|\gs{\vec{b}}_i^{(2)}\|/4 = \frac{3}{4\alpha_i^2}  \|\gs{\vec{b}}_i^{(1)}\| \leq \|\gs{\vec{b}}_{i+1}^{(1)}\|^2/\alpha_i^2 + (\mu_{i,i+1}^{(1)})^2 \|\gs{\vec{b}}_i^{(1)}\|^2/\alpha_i^2 
        =  \frac{\alpha_{i+1}^2}{\alpha_{i}^2}\|\gs{\vec{b}}_{i+1}^{(2)}\|^2 + (\mu_{i,i+1}^{(2)})^2 \|\gs{\vec{b}}_i^{(2)}\|^2
        \; .
    \]
    Now, if $\alpha_i \geq \alpha_{i+1}$, then we immediately see that $\basis^{(2)}$ is LLL-reduced. Otherwise, from \cref{clm:Gram-Schmidts_decay_and_rescaling}, we see that 
    \[
        \|\gs{\vec{b}}_i^{(2)}\| < \frac{2}{\gamma} \|\gs{\vec{b}}_{i+1}^{(2)}\| \leq \|\gs{\vec{b}}_{i+1}^{(2)}\|
    \]
    in which case we trivially have that $3\|\gs{\vec{b}}_i^{(2)}\|^2/4 \leq \|\gs{\vec{b}}_{i+1}^{(2)}\|^2 + (\mu_{i,i+1}')^2 \|\gs{\vec{b}}_i^{(2)}\|^2$. 
    So, $\basis^{(2)}$ is LLL-reduced.

    The fact that 
    \[
    \|\vec{b}_1^{(2)}\|/2^n \leq \|\gs{\vec{b}}_i^{(2)}\|
    \]
    is a direct consequence of \cref{eq:LLL_decay} plus the fact that $\basis^{(2)}$ is LLL-reduced. The fact that 
    \[
        \|\gs{\vec{b}}_i^{(2)}\| \leq 2 \gamma^i \|\vec{b}_1^{(2)}\|
    \]
    follows immediately from \cref{clm:scaled_gram_schmidts} and the definition of $\alpha_i$. And, since $\vec{B}^{(2)}$ is LLL reduced, we have
    \[
        \|\vec{b}_i^{(2)}\|^2 = \|\gs{\vec{b}}_i^{(2)}\|^2 + \sum_{j < i} \mu_{j,i}^2 \|\gs{\vec{b}}_j^{(2)}\|^2 \leq \|\gs{\vec{b}}_i^{(2)}\|^2 + (1/4) \cdot \sum_{j < i} 2^{i-j}\|\gs{\vec{b}}_i^{(2)}\|^2 \leq (4 \gamma^2)^i \|\gs{\vec{b}}_1^{(2)}\|^2
        \; . 
		\qedhere
    \]
\end{proof}

\begin{proposition}
    \label{prop:reducing_to_well_conditioned}
    If $\basis^{(1)} \in \Q^{d \times n}$ is an LLL-reduced basis of $\lat^{(1)}$, $\gamma \geq 2^{2n}$, $\beta > 0$, and $\basis^{(2)} := A_{\basis^{(1)}, \gamma} \basis^{(1)}$. Then, for any $\ell \in [n]$ such that
    \[
    \det(\lat(\vec{b}_1^{(1)},\ldots, \vec{b}_\ell^{(1)})) > n^{2n^2} (\beta/\gamma^{\ell/n}) \cdot  \det(\lat^{(1)})^{\ell/n}
    \]
    and any $\vec{Z} \in \Z^{n \times \ell}$ such that $\basis^{(2)} \vec{Z}$ is a primitive rank-$\ell$ sublattice with $\det(\lat(\basis^{(2)}\vec{Z})) \leq \beta \det(\lat^{(2)})^{\ell/n}$, we must have $\det(\lat(\basis^{(1)}\vec{Z})) \leq \beta \det(\lat^{(1)})^{\ell/n}$.
\end{proposition}
    \begin{proof}  
    Let $k_-$ be maximal such that $\vec{b}_1^{(2)},\ldots, \vec{b}_{k_-}^{(2)} \in \lat(\basis^{(2)} \vec{Z})$ (where we take $k_- = 0$ if $\vec{b}_1^{(2)} \notin \lat(\basis^{(2)} \vec{Z})$), and let $k_+$ be minimal such that $\lat(\basis^{(2)} \vec{Z}) \subseteq \lat(\vec{b}_1^{(2)},\ldots, \vec{b}_{k_+}^{(2)})$. Notice that we also have $\lat(\vec{b}_1^{(1)},\ldots, \vec{b}_{k_-}^{(1)}) \subseteq \lat(\basis^{(1)} \vec{Z}) \subseteq \lat(\vec{b}_1^{(1)},\ldots, \vec{b}_{k_+}^{(1)})$. 
    
    We define $\lat_-^{(i)} := \lat(\vec{b}_1^{(i)},\ldots, \vec{b}_{k_-}^{(i)})$ and $\lat_+^{(i)} := \Pi_{(\lat_-^{(i)})^\perp}(\lat(\basis^{(i)} \vec{Z}))$. We have
    \[
        \det(\lat(\basis^{(i)} \vec{Z})) = \det(\lat_-^{(i)}) \cdot \det(\lat_+^{(i)}))
        \; .
    \]
    And, notice that
    \[
        \det(\lat_-^{(1)}) = \det(\lat_-^{(2)}) \cdot \prod_{i=1}^{k_-} \alpha_i
        \; .
    \]
    Similarly, since 
    $
        \lat_+^{(i)} \subseteq \Pi_{\{\vec{b}_1^{(1)},\ldots, \vec{b}_{k_-}^{(1)}\}^\perp}(\lat(\vec{b}_{k_-+1}^{(i)},\ldots, \vec{b}_{k_+}^{(i)})
        \; 
    $, 
    we see that
    $
        \det(\lat_+^{(1)}) \leq \alpha_{k_+}^{\ell - k_-} \det(\lat_+^{(2)})
    $. 
    On the other hand, we have that
    \[
        \det(\lat^{(1)}) = \det(\lat^{(2)})  \cdot \prod_{i=1}^n \alpha_i
        \; .
    \]
    Therefore,
    \begin{align*}
        \Delta 
            &:= \frac{\det(\lat(\basis^{(1)} \vec{Z}))}{\det(\lat(\basis^{(2)} \vec{Z}))} \cdot \frac{\det(\lat^{(2)})^{\ell/n}}{\det(\lat^{(1)})^{\ell/n}} \\
            &= \frac{\det(\lat_+^{(1)})}{\det(\lat_+^{(2)})} 
 \cdot \prod_{i=1}^{k_-} \alpha_i \cdot \prod_{i=1}^n \alpha_{i}^{-\ell/n}
 \\
 &\leq \alpha_{k_+}^{\ell-k_-} \alpha_{k_-}^{(1-\ell/n) k_-} \alpha_{k_-+1}^{-\ell (n-k_-)/n} \\
 &\leq \Big( \frac{\alpha_{k_+}}{\alpha_{k_-+1}}\Big)^{\ell - k_-}
 \; .
    \end{align*}
    We wish to prove that $\Delta \leq 1$. If $k_+ \leq k_-+1$, then we are clearly done since the $\alpha_i$ are non-decreasing. So, we may assume that $k_+ > k_-+1$.

    Then, let $k \in \{k_-+1,k_-+2,\ldots, k_+-1\}$. If $\alpha_{k} \neq \alpha_{k+1}$, then by \cref{clm:Gram-Schmidts_decay_and_rescaling}, this implies that $\|\gs{\vec{b}}_k^{(2)}\| < (2/\gamma) \|\gs{\vec{b}}_{k+1}^{(2)}\|$ and by \cref{eq:LLL_decay}, we see that for all $i \leq k < j$ we must have that 
    \[\|\gs{\vec{b}}_i^{(2)}\| < (2^{n}/\gamma) \|\gs{\vec{b}}_{j}^{(2)}\|
    \; .
    \]
    Set $\alpha := \gamma/2^n$.
    If $k \geq \ell$, then by \cref{lem:dense_sublattice_gram_schmidts} since $\lat(\basis^{(2)} \vec{Z}) \not \subseteq \lat(\vec{b}_1^{(2)},\ldots, \vec{b}_k^{(2)})$, we must have that
    \[
         \min_{|S|=\ell} \prod_{i \in S} \|\gs{\vec{b}}_i^{(2)}\| \leq (\sqrt{n}/\alpha) \det(\lat(\basis^{(2)} \vec{Z})) \leq (\sqrt{n}/\alpha) \beta \det(\lat^{(2)})^{\ell/n}
        \; .
    \]
    But, by \cref{eq:LLL_decay}, we see that
    \[
        \prod_{i=1}^\ell \|\gs{\vec{b}}_i^{(2)}\| \leq 2^{\ell n} \min_{|S|=\ell} \prod_{i \in S} \|\gs{\vec{b}}_i^{(2)}\| \leq (\sqrt{n}/\alpha) \beta \det(\lat^{(2)})^{\ell/n}
        \; .
    \]
    And, it follows that
    \[
        \prod_{i=1}^\ell \|\gs{\vec{b}}_i^{(1)}\| = \prod_{i=1}^\ell (\alpha_i \|\gs{\vec{b}}_i^{(2)}\|) \leq (\sqrt{n}/\alpha) \beta \cdot \det(\lat^{(2)})^{\ell/n} \cdot \prod_{i=1}^\ell \alpha_i \leq (\sqrt{n}/\alpha) \beta \cdot \det(\lat^{(2)})^{\ell/n} \cdot \prod_{i=1}^n \alpha_i^{\ell/n} = (\sqrt{n}/\alpha) \beta \cdot \det(\lat^{(1)})^{\ell/n} 
        \; ,
    \]
    a contradiction.
    If $k \leq \ell$, then by \cref{cor:other_dense_sublattice_gram_schmidt_thing} we similarly must have
    \[
    \prod_{i=1}^k \|\gs{\vec{b}}_i^{(2)}\|^{\ell/k} \leq  (\alpha/n^n)^{-\ell/k} \cdot \det(\lat(\basis^{(2)} \vec{Z})) \leq \beta \cdot  (\alpha/n^n)^{-\ell/k} \cdot \det(\lat^{(2)})^{\ell/n}
        \; .
    \]
    And, we similarly notice that by \cref{eq:LLL_decay}
    \[
        \prod_{i=1}^\ell \|\gs{\vec{b}}_i^{(2)} \| \leq 2^{\ell n} \prod_{i=1}^k \|\gs{\vec{b}}_i^{(2)} \|^{\ell/k} \leq 2^{\ell n} \beta (\alpha/n^n)^{-\ell/k} \det(\lat^{(2)})^{\ell/n}
        \;,
    \]
    and we similarly derive a contradiction. So, we conclude that $\alpha_{k_+} = \alpha_{k_-+1}$ and therefore $\Delta \leq 1$, as needed.
\end{proof}

\subsection{An approximate representation}
\label{app:bounded_bits}

In this section, we prove that, up to an arbitrarily small loss in the approximation factor, we may always assume that the input lattice $\lat$ to our DSP algorithm has a succinct representation.

Our main remaining technical tool is the following.

\begin{claim}
    \label{clm:det_deriv}
    Let $\vec{V} = (\vec{v}_1,\ldots, \vec{v}_\ell), \vec{W} = (\vec{w}_1,\ldots, \vec{w}_\ell) \in \R^{d \times \ell}$ be two matrices with 
    \[
        \|\vec{V} - \vec{W}\|_\infty \leq \eps \min_i \|\vec{w}_i\|^2/\|\vec{W}\|_\infty
        \; 
    \]
    for some $\eps < 1/(10\ell^2)$. Then, 
    \[
        |\det(\vec{V}^T \vec{V})  - \det(\vec{W}^T \vec{W})| \leq  20\eps \ell^{2\ell}  \prod_{i=1}^\ell \|\vec{w}_i\|^2
        \; .
    \]
\end{claim}
\begin{proof}
    Let $\vec{E} := \vec{V} - \vec{W}$. We have
    \[
        \det(\vec{V}^T \vec{V}) = \det(\vec{W}^T \vec{W} + \vec{E}^T \vec{W} + \vec{W}^T \vec{E} + \vec{E}^T \vec{E})
        \; .
    \]
    Let $\vec{G} := (\vec{g}_1,\ldots, \vec{g}_\ell) := \vec{W}^T \vec{W}$ and $\vec{H} := (\vec{h}_1,\ldots, \vec{h}_\ell) := \vec{E}^T \vec{W} + \vec{W}^T \vec{E} + \vec{E}^T \vec{E}$. Then,
    \begin{align*}
        |\det(\vec{V}^T \vec{V})  - \det(\vec{W}^T \vec{W})|
            &= \left|\sum_{\sigma \in S_\ell} (-1)^{\mathrm{sign}(\sigma)}\prod_{i=1}^\ell (G_{i,\sigma(i)} + H_{i, \sigma(i)}) - \det(\vec{G})\right| \\
            &= \left| \sum_{\sigma \in S_\ell} (-1)^{\mathrm{sign}(\sigma)} \sum_{T \subseteq [\ell]} \prod_{i \in T} G_{i,\sigma(i)} \cdot \prod_{i \notin T} H_{i, \sigma(i)}- \det(\vec{G})\right| \\
            &= \left| \sum_{\sigma \in S_\ell} (-1)^{\mathrm{sign}(\sigma)} \sum_{T \subsetneq [\ell]} \prod_{i \in T} G_{i,\sigma(i)} \cdot \prod_{i \notin T} H_{i, \sigma(i)}\right| \\
            &\leq \ell! \cdot \sum_{T \subsetneq [\ell]} \prod_{i \in T} \|\vec{g}_i\| \cdot \prod_{i \notin T} \|\vec{h}_i\| \\
            &\leq  2^\ell \ell!\cdot  \prod_{i=1}^\ell \|\vec{g}_i\| \cdot \max_i \frac{\|\vec{h}_i\|}{\|\vec{g}_i\|}
            \; ,
    \end{align*}
    where we have used the fact that $\|\vec{h}_i\|/\|\vec{g}_i\| \leq 1$ for all $i$ (so that the summand in the second-to-last expression is largest when $|T| = \ell-1$).
    Finally, we notice that $\|\vec{g}_i\| \geq \|\vec{w}_i\|^2$ and 
    \[
        \|\vec{h}_i\| \leq \|\vec{E} \vec{w}_i\| + \|\vec{W} \vec{e}_i\| + \|\vec{E} \vec{e}_i\| \leq 3\ell \|\vec{E}\|_\infty \|\vec{W}\|_\infty \leq 3\eps \ell \min_j \|\vec{w}_j\|^2
        \; .
    \]
    It follows that 
    \[
        |\det(\vec{V}^T \vec{V}) - \det(\vec{W}^T \vec{W})| \leq  3\eps \ell 2^\ell \ell!  \prod_{i=1}^\ell \|\vec{g}_i\| \leq  20\eps \ell^{2\ell}  \prod_{i=1}^\ell \|\vec{w}_i\|^2
        \; ,
    \]
    as needed.
\end{proof}

\begin{theorem}
    \label{thm:rounded_basis}
    There exists a ``basis rounding algorithm'' with the following property. It takes as input a basis $\basis \in \Q^{d \times n}$ for a lattice $\lat$, an $\ell \in [n]$, and a size parameter $m' \geq 10n^5$. It outputs a basis $\basis'$ for a ``rounded'' lattice $\lat'$ such that (1) $\basis'$ has bitlength at most $m'$ (when represented as a list of numerators and denominators); and (2) for any rank-$\ell$ sublattice $\lat'' \subseteq \lat'$ such that $\gamma := \det(\lat'')/\det(\lat')^{\ell/n} \leq 2^{n^2}$, we have
    \[
        \det( \widehat{\lat}) \leq \gamma' \cdot \det(\lat)^{\ell/n}
        \; ,
    \]
    where  $\widehat{\lat} := \basis (\basis')^{-1}\lat''$ and
    \[
        \gamma' := \max\{1, (1+2^{3n^5}/2^{m'})\gamma \}
        \; .
    \]
    Furthermore, the algorithm runs in time $O( (mm')^5)$, where $m$ is the bit length of the input lattice.
\end{theorem}
\begin{proof}
    We may assume without loss of generality that $n \geq 25$.

    On input $\basis \in \{-p/q,\ldots, p/q\}^{d \times n}$ of a lattice $\lat$, $\ell \in [n]$, and $m'$, the algorithm $\mathcal{A}$ behaves as follows, where $M := 2^{m'}/2^{2n^5} \geq 2^{n^5}$.
        \begin{enumerate}
            \item It computes an LLL-reduced basis $\basis^{(1)}$ of $\lat$. 
            \item If $\det(\vec{b}_1^{(1)},\ldots, \vec{b}_\ell^{(1)}) \leq  \det(\lat)^{\ell/n}$, then it simply outputs the diagonal matrix 
            \[\basis' := (\vec{e}_1,\vec{e}_2,\ldots, \vec{e}_\ell, M \vec{e}_{\ell+1},\ldots, M \vec{e}_n)
            \; .
            \]
            \item Otherwise, it sets $\basis^{(2)} := A_{\basis^{(1)}, n^{2n^3}} \basis^{(1)}/R$, where $R := \|\vec{b}_1^{(1)}\|/M$ and then outputs $\basis' := \floor{\basis^{(2)}}$ , where $\floor{\cdot}$ represents rounding each coordinate towards zero.
        \end{enumerate}

    First, notice that the algorithm trivially has the desired properties in the corner case when $\det(\vec{b}_1^{(1)},\ldots, \vec{b}_\ell^{(1)}) \leq \det(\lat)^{\ell/n}$. (E.g., by \cref{lem:dense_sublattice_gram_schmidts}, it follows that the only sublattice of $\lat'$ with determinant bounded by $2^{n^2} \det(\lat')^{\ell/n}$ in this case is the sublattice generated by $\vec{b}_1',\ldots, \vec{b}_\ell'$.) So, we ignore this case below.

    Suppose that $\vec{Z} \in \Z^{n \times \ell}$ is such that $\basis' \vec{Z}$ is an LLL-reduced basis of $\lat''$. (Notice that this implies that $\basis \vec{Z}$ is a basis of $\widehat{\lat}$, though not necessarily an LLL-reduced basis. We are only using the existence of such a $\vec{Z}$ for the analysis of the algorithm---it is not computed by the algorithm itself.) Before analyzing the algorithm directly, we prove some basic bounds on the norms of the various matrices in question, which will be useful throughout our analysis. In particular, by 
    \cref{cor:LLL_after_rescaling}, we see that $\basis^{(2)} := (\vec{b}_1^{(2)},\ldots, \vec{b}_n^{(2)})$ satisfies that 
    \begin{equation}
        \label{eq:B_2_well_conditioned}
        M \leq \|\vec{b}_i^{(2)}\| \leq n^{3n^4}  M
        \; .
    \end{equation} 
    And, since $\basis' \vec{Z}$ is an LLL-reduced basis of $\lat''$, we see from \cref{thm:LLL_orthogonality_defect} that
    \[
        \prod_i \|\basis' \vec{z}_i\| \leq 2^{\ell^2/4} \det(\lat'') \leq \gamma 2^{\ell^2/4} \cdot \det(\lat')^{\ell/n} \leq n^{4n^4}  M^\ell
        \; .
    \]
    Furthermore, for any $\vec{z} \in \Z^n$, letting $\vec{E} := \vec{B}^{(2)} - \vec{B}'$, we see that 
    \begin{equation}
        \label{eq:lambda_1_in_B_dagger_lower_bound}
        \|\basis' \vec{z}\| \geq \|\basis^{(2)} \vec{z}\| - \|\vec{E} \vec{z}\| \geq \|\vec{b}_1^{(2)}\| \|\vec{z}\|/3^{n/2} - \sqrt{n} \|\vec{z}\| \geq M \|\vec{z}\|/3^{n}
        \; ,
    \end{equation}
    where the second inequality is \cref{lem:LLL_distortion} (together with the fact that $\|\vec{E}\|_\infty \leq 1$) and the last inequality uses \cref{eq:B_2_well_conditioned} and the lower bound on $M$.
    Combining the above two inequalities (together with the simple fact that $\|\vec{z}_j\| \geq 1$ for all $j$) gives that for each $i$,
    \begin{equation}
        \label{eq:bound_on_z_i}
        \|\vec{z}_i\| \leq 3^{n}  \|\basis' \vec{z}_i\|/M \leq n^{5n^4} M^{-(\ell-1)} \cdot \prod_{j \neq i} \|\basis' \vec{z}_j\|^{-1} \leq n^{6n^4}
        \; .
    \end{equation}
    Finally, by \cref{cor:LLL_bounded_bit_length}, we have that $\|\vec{B}^{(1)}\|_\infty \leq 2^{4n^2 + 3nm}$. Therefore, 
    \begin{equation}
        \label{eq:final_basis_bound_on_bits}
         \|\vec{B}^{(1)} \vec{Z}\|_\infty \leq n\|\vec{B}^{(1)}\|_\infty \|\vec{Z}\|_\infty \leq n^{7n^4} 2^{3nm}
         \; .
    \end{equation}

    Now, we observe that the running time and outputs size of the algorithm are as claimed. In particular, the bases $\basis^{(1)}$, $\basis^{(2)}$, and $\basis'$ can be computed in time $O(m' m)^5$ by \cref{thm:LLL_with_bit_lengths}. Furthermore, by \cref{eq:B_2_well_conditioned}, we see that $\vec{B}' \in \{-2^{n^5}M ,\ldots, 2^{n^5}M\}^{d \times n}$. It follows that $\vec{B}'$ has bitlength bounded by $m'$, as needed.

    Finally, we prove the bound on $\det(\widehat{\lat})$. By \cref{prop:reducing_to_well_conditioned}, it suffices to show that 
    \[
        \det(\lat(\basis^{(2)} \vec{Z})) \leq \gamma' \det(\lat(\basis^{(2)}))^{\ell/n}
        \; .
    \]
    
    To that end, we first recall that we trivially have $\|\basis' - \basis^{(2)}\|_\infty \leq 1$. Together with \cref{eq:B_2_well_conditioned}, it follows that
    \[
        \|\basis' - \basis^{(2)}\|_\infty \leq \eps \min_i \|\vec{b}_i^{(2)}\|^2/\|\vec{B}^{(2)}\|_\infty
        \; ,
    \]
    where $\eps := n^{3n^4}/M$. It then follows from \cref{clm:det_deriv} that
    \begin{align}
        \det(\lat')^2 
            &\leq \det(\lat^{(2)})^2 + 20 \eps n^{2n} \prod_{i=1}^n \|\vec{b}_i^{(2)}\|^2 \nonumber \\
            &\leq (1+20 \eps 2^{n^2}) \cdot \det(\lat^{(2)})^2 \nonumber \\
            &\leq (1+n^{4n^4}/M) \det(\lat^{(2)})^2 \label{eq:det_B_dagger_det_B_2}
        \; ,
    \end{align}
    where the second inequality uses \cref{thm:LLL_orthogonality_defect} (using the fact that $\basis^{(2)}$ is LLL-reduced).

    Next, let $\vec{W} := \basis^{(2)} \vec{Z}$ and $\vec{V} := \basis' \vec{Z}$. Then,
    \[
        \|\vec{V} - \vec{W}\|_\infty \leq n\|\vec{B}^{(2)}\|_\infty \|\vec{Z}\|_\infty \leq n^{7n^4}
        \; ,
    \]
    where we have applied \cref{eq:bound_on_z_i}. Furthermore, we have $\|\vec{V}\|_\infty \leq n^{10n^4}M$ (using the fact that $\|\basis'\|_\infty \leq n^{3n^4}M$ and $\|\vec{Z}\|_\infty \leq n^{6n^4}$) and from \cref{eq:lambda_1_in_B_dagger_lower_bound} we have that $\min_i \|\vec{v}_i\| \geq M/3^{n}$. Putting these things together, we see that
    \[
        \|\vec{V} - \vec{W}\|_\infty \leq \eps' \min_i \|\vec{v}_i\|^2/\|\vec{V}\|_\infty
        \; ,
    \]
    where $\eps' := n^{11n^4}/M$. It follows from \cref{clm:det_deriv} that
    \[
        \det(\lat(\basis^{(2)} \vec{Z}))^2 = \det(\lat'')^2 + 20 \eps' \ell^{2\ell} \prod_{i=1}^\ell \|\vec{v}_i\|^2 \leq \det(\lat'')^2 + \frac{n^{12n^4}}{M} \cdot \prod_{i=1}^\ell \|\vec{v}_i\|^2
        \; .
    \]
    Finally, we note that since $\vec{V} = \basis' \vec{Z}$ is an LLL-reduced basis of $\lat''$, we have by yet another application of \cref{thm:LLL_orthogonality_defect} that
    \[
        \prod_{i=1}^\ell \|\vec{v}_i\|^2 \leq 2^{n^2} \det(\lat'')^2
        \; .
    \]
    Therefore
    \[
        \det(\lat(\basis^{(2)} \vec{Z}))^2 \leq (1+ n^{15n^4}/M) \det(\lat'')^2
        \; .
    \]

    Finally, combining the definition of $\gamma$, which tells us that $\det(\lat'') \leq \gamma \det(\lat')^{\ell/n}$ and \cref{eq:det_B_dagger_det_B_2}, we see that
    \begin{align*}
        \det(\lat(\basis^{(2)} \vec{Z}))^2
            &\leq (1+ n^{15n^4}/M) \cdot \gamma^2 \det(\lat')^{2\ell/n} \\
            &\leq (1+ n^{15n^4}/M)(1+n^{4n^4}/M)  \cdot \gamma^2 \det(\lat^{(2)})^{2\ell/n} \\
            &\leq (\gamma')^2 \det(\lat^{(2)})^{2\ell/n}
        \; ,
    \end{align*}
    as needed.
\end{proof}

\subsection{On related works and techniques in basis reduction}
\label{sec:related_work_app}

The literature on optimization techniques for bit-level representation of lattices within basis reduction algorithms is extensive and remains an active area of research, particularly following the advent of floating-point arithmetic instantiations of LLL and BKZ algorithms. The primary technical point concerns the stability of the reduction when ``rounding'' the basis, i.e., when considering only a certain fraction of the bit representation of the coefficients. To some extent, the results presented in this appendix relate to these works, and we aim to draw connections between our results for DSP and related works for SVP.

Our proof technique for \cref{cor:other_dense_sublattice_gram_schmidt_thing}, which demonstrates how to retrieve specific vectors of a sublattice when the sequence of Gram-Schmidt coefficients exhibits gaps, is related to Lemma 3.3 of \cite{DucasWoerdenNTRU}, which bounds the determinant of an intersection of a sublattice with given basis vectors.

Our main technical tool, defined in \cref{app:well-conditioned}, provides a method to ``compress'' a lattice basis by reducing the size of the gaps in its profile by scaling each vector accordingly. This is quite similar to \cite[Definition 5]{RHFastPracticalLattice2023}. In our case, we examine the gap relative to the length of the first vector, whereas \cite{RHFastPracticalLattice2023}  deals with the relative growth of gaps. Although conceptually similar, this difference prevents us from applying their analysis for LLL \emph{verbatim}.

In \cref{cor:LLL_after_rescaling}, we study the stability of LLL-reduced bases under rounding, which is analogous to the work in \cite{SMSVLLLReducingMost2014}. The main difference here is that the rounding we consider is more nuanced than simply dropping the least significant bits, as it is adapted to the size of gaps in the profile. Consequently, Lemma 13 of \cite{SMSVLLLReducingMost2014} cannot be directly applied to prove our result.

\subsection{A note on computing the dual and the intersection}
\label{app:computing_stuff}

Finally, we address how to efficiently compute a basis $\basis^*$ for the dual $\lat^*$ of a lattice $\lat$, given a basis $\basis \in \Q^{d \times n}$ for $\lat$ and (more interestingly) how to efficiently compute a basis $\basis''$ for $\lat \cap (\lat')^\perp$ given a basis $\basis  \in \Q^{d \times n}$ for $\lat$ and a basis $\basis' \in \Q^{d \times \ell^*}$ for a sublattice $\lat' \subseteq \lat^*$ of the dual.

Computing the dual basis $\basis^*$ can be done via the explicit formula
\[
    \basis^* = \basis (\basis^T \basis)^{-1}
    \; .
\]
This can of course be computed in time that is polynomial in the bitlength $m$ of $\basis$, e.g., $O(m^3)$ time suffices.

To compute a basis for the intersection  $\lat \cap (\lat')^\perp$ efficiently, we rely on \cref{lemma:noah}, which after taking duals tells us that
\[
(\lat \cap (\lat')^\perp)^* =
 \Pi_{(\lat')^\perp}(\lat^*)
 \; .
\]
So, it suffices to show how to efficiently compute a basis $\widehat{\basis}$ for $\widehat{\lat} := \Pi_{(\lat')^\perp}(\lat^*)$ (as we can then compute $\basis'' := \widehat{\basis}^*$). 

For this, we first compute $\Pi_{(\lat')^\perp}(\basis^*) \in \Q^{d \times n}$. This gives us $n$ vectors that clearly generate $\widehat{\lat}$, but they are not typically a basis (because a basis for this lattice with rank $n-\ell^*$ must consist of $n-\ell^*$ linearly independent vectors). We then can use any efficient algorithm for converting such a generating set to an actual basis basis. See, for instance \cite[Table 1]{LNComputingLatticeBasis2019} for a list of many different algorithms for this problem. This can be done, e.g., in time $O(m^5)$.
\end{document}